\documentclass{article}

\usepackage{typearea}
\usepackage{setspace}

\usepackage{fullpage}

\usepackage{comment} 
\usepackage{bbm}
\usepackage{framed} 
\usepackage{url} 
\usepackage{complexity}
\usepackage{booktabs}
\usepackage{amsmath,amssymb}
\usepackage{float}

\usepackage{amsthm}
\usepackage{thmtools} 
\usepackage{thm-restate}
\usepackage{nicefrac}
\usepackage{calc}
\usepackage{enumerate}
\usepackage{enumitem}
\usepackage[usenames,dvipsnames]{xcolor}

\usepackage[ruled,vlined,linesnumbered]{algorithm2e}
\usepackage{forloop}

\usepackage{graphicx}
\usepackage[font=footnotesize,labelfont=bf]{subcaption}
\usepackage[font=footnotesize,labelfont=bf]{caption}
\usepackage[nobreak=true]{mdframed}
\usepackage{appendix}
\usepackage[noend]{algpseudocode}
\usepackage[colorinlistoftodos]{todonotes}
\usepackage{xr}
\usepackage{array}
\usepackage{xspace}
\definecolor{ForestGreen}{rgb}{0.1333,0.5451,0.1333}
\definecolor{DarkRed}{rgb}{0.8,0,0}
\definecolor{Red}{rgb}{1,0,0}
\usepackage[linktocpage=true,
pagebackref=true,colorlinks,
linkcolor=DarkRed,citecolor=ForestGreen,
bookmarks,bookmarksopen,bookmarksnumbered]{hyperref}

\usepackage[capitalise]{cleveref}
\crefrangelabelformat{enumi}{#3#1#4--#5#2#6}
\usepackage{parskip}
\usepackage{chngcntr}
\usepackage{mathtools,stackengine}
\usepackage{multirow}

\stackMath
\newcommand{\stackGeq}[1]{%
	\setbox0=\hbox{${}\mathrel{\stackon[-1pt]{\geq}{\scriptstyle\text{#1\strut}}}{}$}
	\xdef\tmpwd{\dimexpr\the\wd0\relax}
	\kern.5\tmpwd\mathclap{\box0}&\kern.5\tmpwd
}

\usepackage{nameref}

\usepackage{tikz}
\usetikzlibrary{patterns}

\allowdisplaybreaks

\newcommand\Z{\mathbb{Z}}

\newcommand\eps{\epsilon}

\DeclarePairedDelimiterX{\expectarg}[1]{[}{]}{%
	\ifnum\currentgrouptype=16 \else\begingroup\fi
	\activatebar#1
	\ifnum\currentgrouptype=16 \else\endgroup\fi
}

\newcommand\tbuy{\mathcal{T}_{\text{cust}}}
\newcommand\tsell{\mathcal{T}_{\text{supp}}}

\newcommand\primy{\overline{y}}
\newcommand\primz{\overline{z}}
\newcommand\primr{\overline{r}}

\newcommand\dualx{\overline{x}}
\newcommand\duall{\overline{\ell}}
\newcommand\duala{\overline{\alpha}}
\newcommand\dualb{\overline{\beta}}

\newcommand\vt{v}

\DeclarePairedDelimiterX{\nicesetarg}[1]{\{}{\}}{%
	\ifnum\currentgrouptype=16 \else\begingroup\fi
	\activatebar#1
	\ifnum\currentgrouptype=16 \else\endgroup\fi
}

\newcommand{\innermid}{\nonscript\;\delimsize\vert\nonscript\;}
\newcommand{\activatebar}{%
	\begingroup\lccode`\~=`\|
	\lowercase{\endgroup\let~}\innermid 
	\mathcode`|=\string"8000
}

\newcommand\alg{\textsc{Alg}\xspace}

\counterwithin{equation}{section}

\usepackage{eqparbox}

\newcommand\wKL[2]{\textsc{KL}_w\left(#1 \mid \mid #2\right)}

\theoremstyle{plain}

\newtheorem{theorem}{Theorem}[section]

\newtheorem{lemma}[theorem]{Lemma}

\newtheorem{observation}[theorem]{Observation}
\newtheorem{claim}[theorem]{Claim}

\newtheorem{invariant}{Invariant}

\newtheorem{assumption}[theorem]{Assumption}
\newtheorem{remark}[theorem]{Remark}

\newlength{\continueindent}
\usepackage{etoolbox}
\makeatletter
\newcommand*{\ALG@customparshape}{\parshape 2 \leftmargin \linewidth \dimexpr\ALG@tlm+\continueindent\relax \dimexpr\linewidth+\leftmargin-\ALG@tlm-\continueindent\relax}
\apptocmd{\ALG@beginblock}{\ALG@customparshape}{}{\errmessage{failed to patch}}
\makeatother

\makeatletter
\def\thm@space@setup{%
	\thm@preskip=\parskip \thm@postskip=0pt
}
\makeatother

\usepackage{etoolbox}
\usepackage{tikz}
\usetikzlibrary{tikzmark}
\usetikzlibrary{calc}

\errorcontextlines\maxdimen

\newcommand{\ALGtikzmarkcolor}{black}
\newcommand{\ALGtikzmarkextraindent}{4pt}
\newcommand{\ALGtikzmarkverticaloffsetstart}{-.5ex}
\newcommand{\ALGtikzmarkverticaloffsetend}{-.5ex}
\makeatletter
\newcounter{ALG@tikzmark@tempcnta}

\newcommand\ALG@tikzmark@start{%
	\global\let\ALG@tikzmark@last\ALG@tikzmark@starttext%
	\expandafter\edef\csname ALG@tikzmark@\theALG@nested\endcsname{\theALG@tikzmark@tempcnta}%
	\tikzmark{ALG@tikzmark@start@\csname ALG@tikzmark@\theALG@nested\endcsname}%
	\addtocounter{ALG@tikzmark@tempcnta}{1}%
}

\def\ALG@tikzmark@starttext{start}
\newcommand\ALG@tikzmark@end{%
	\ifx\ALG@tikzmark@last\ALG@tikzmark@starttext
	\else
	\tikzmark{ALG@tikzmark@end@\csname ALG@tikzmark@\theALG@nested\endcsname}%
	\tikz[overlay,remember picture] \draw[\ALGtikzmarkcolor] let \p{S}=($(pic cs:ALG@tikzmark@start@\csname ALG@tikzmark@\theALG@nested\endcsname)+(\ALGtikzmarkextraindent,\ALGtikzmarkverticaloffsetstart)$), \p{E}=($(pic cs:ALG@tikzmark@end@\csname ALG@tikzmark@\theALG@nested\endcsname)+(\ALGtikzmarkextraindent,\ALGtikzmarkverticaloffsetend)$) in (\x{S},\y{S})--(\x{S},\y{E});%
	\fi
	\gdef\ALG@tikzmark@last{end}%
}

\apptocmd{\ALG@beginblock}{\ALG@tikzmark@start}{}{\errmessage{failed to patch}}
\pretocmd{\ALG@endblock}{\ALG@tikzmark@end}{}{\errmessage{failed to patch}}
\makeatother

\algblock[with]{With}{EndWith}
\algblockdefx[With]{With}{EndWith}%
[1]{\textbf{with} #1 \textbf{do}}%
{}

\makeatletter
\ifthenelse{\equal{\ALG@noend}{t}}%
{\algtext*{EndWith}}
{}%
\makeatother

\title{Competitive Bundle Trading}

\author{
 {Yossi Azar\thanks{Department of Computer Science, Tel-Aviv University,
     Tel-Aviv, Israel. Emails: {\tt azar@tauex.tau.ac.il, orvardi@mail.tau.ac.il}.}
    }
 \and
 {Niv Buchbinder\thanks{Department of Statistics and Operations Research, School of Mathematical Sciences, Tel Aviv University, Tel Aviv, Israel. Email: \texttt{niv.buchbinder@gmail.com}.}
 }
 \and
 {Roie Levin\thanks{Department of Computer Science, Rutgers University, Piscataway, NJ 08854. Email: {\tt roie.levin@rutgers.edu.}}}
 \and
 {Or Vardi$^{*}$
   } 
}    

\date{}

\begin{document}
\maketitle 
\begin{abstract}
A retailer is purchasing goods in bundles from suppliers and then selling these goods in bundles to customers; her goal is to maximize profit, which is the revenue obtained from selling goods minus the cost of purchasing those goods. In this paper, we study this general trading problem from the retailer's perspective, where both suppliers and customers arrive online. The retailer has inventory constraints on the number of goods from each type that she can store, and she must decide upon arrival of each supplier/customer which goods to buy/sell in order to maximize profit.

We design an algorithm with logarithmic competitive ratio compared to an optimal offline solution. We achieve this via an exponential-weight-update dynamic pricing scheme, and our analysis dual fits the retailer's profit with respect to a linear programming formulation upper bounding the optimal offline profit. We prove (almost) matching lower bounds, and we also extend our result to an incentive compatible mechanism. Prior to our work, algorithms for trading bundles were known only for the special case of selling an initial inventory.
\end{abstract}

\section{Introduction}

Consider the following online decision-making task faced by a typical retailer. The retailer wishes to purchase $n$ types of goods in bundles from suppliers (e.g. CPUs, Hard drives, NVIDIA graphics cards, etc.) and subsequently sell these goods in customized bundles to customers (e.g. desktop computers). For each item type $i\in [n]$ there is some inventory limit, $w_i$, on the amount that can be stored at any given time. The goal is to maximize profit, i.e. revenue obtained from sales minus cost paid for purchases. What is more, the retailer does not know what the market will look like in the future. Instead, customers and suppliers valuations for goods may change unpredictably over time, and the retailer must do their best to adapt.

To model this scenario concretely, we imagine that at every time step $t$, either a supplier or a customer arrives with a menu of bundles $S^t$, where each bundle $s\in S^t$ is a (multi-)set of the item types and has a value $v^t_s$ to the supplier/customer. In the supplier case, the retailer has the option to \emph{buy} one of the bundles from the menu at a price $v^t_s$, and in the customer case, the retailer has the option to \emph{sell} one of the bundles from the menu at a price of $v^t_s$. We call this the \emph{online trading problem} with \emph{known valuations}. Our problem remains challenging even in the special case where customers/suppliers are {\em single minded}, meaning that each customer would like to purchase a single bundle of items $s$ (or any superset thereof), and similarly, a supplier would like to sell to the retailer a single bundle $s$ (or any of its subsets). 

A classical special case in online algorithms is the problem of selling an initial inventory to customers \cite{AAP93, LM99} (originally motivated by bandwidth allocation in networks, see additional discussion in the next subsection). The problem we consider in this paper is significantly more complicated for a few reasons. For one, here the number of transactions can be arbitrarily large, whereas in the sell-only case, the number of transactions is limited by the initial number of items. Second, our goal function is profit, i.e. difference between sales revenue and purchasing cost. Objective functions of the form $\max_x A(x) - B(x)$ are often already difficult to handle offline with full information.
In the online setting, if the algorithm is not careful, it may even end up with a negative profit. We are not aware of any competitive online algorithms for problems of this form. 

Yet the story is more complicated still. In many contexts, customers and suppliers are strategic agents who do not wish to reveal their true valuations of goods. In this case we would like to design an {\em incentive compatible} mechanism. We refer to this harder variant as the \emph{unknown valuation setting}. We define our problem formally in \cref{sec:prelim}.

\subsection{Our results}
Our main result is logarithmic competitive algorithms for the online trading problem for both the known and unknown valuation settings against optimal offline solutions. Let $d$ be the maximum size of a bundle of any \emph{customer} and let $v$ be the maximum to minimum ratio of the value of any bundle of a \emph{customer}.\footnote{The maximum \emph{supplier} bundle size/value do not play a role in our bounds.} We assume these values (or upper bounds on them) are known to the algorithm. 

\begin{theorem}[Main Theorem, informal]\label{thm:main-informal}
For every instance of the online trading problem with a demand oracle for the valuations, large enough inventory compared with the bundle sizes and every $\eps>0$, there is an 
\begin{itemize}
    \item $O\left(\frac{1}{\eps} \log(dv)\right)$-competitive deterministic algorithm for the \textbf{known valuation} case. 
    \item $O\left(\frac{1}{\eps} \log(\frac{dv}{\eps})\right)$-competitive randomized incentive compatible online algorithm for the \textbf{unknown valuation} case.
\end{itemize}

The competitive ratio of both algorithms is with respect to an optimal offline fractional solution, where supplier values at any time step
are $(1+\eps)$ larger.
\end{theorem}

We place no restrictions on the valuation functions of the customers/suppliers except that these have access to their {\em demand oracle}: given prices for the items, customers/suppliers are able to output a bundle maximizing their utility.\footnote{We remark that while answering a demand query might be NP-hard in many cases, in our context it is natural to expect the customers/suppliers to be able to answer a demand query. Otherwise, it would be unreasonable to expect a mechanism to satisfy their demands.}
On the other hand we need a few other key assumptions that we show are necessary to obtain our result. 
\begin{enumerate}
    \item The optimal offline solution to which we compare our online algorithm's profit sees the same sequence of customers/suppliers as the online algorithm, except its suppliers' valuations are $(1+\eps)$ higher for some fixed $\eps>0$.\footnote{Instead of using a $(1+\eps)$ value augmentation for the suppliers, we could also use a $(1+\eps)$ reduction in the customers' values and we would obtain similar results. We have arbitrarily chosen the first option.} 
    Resource augmentation is a common assumption for many online problems (see, e.g., \cite{R2020}, Chapter 4 for a survey), and in our case it is necessary: our lower bounds, which we describe soon, show that without this assumption, no competitive algorithms exist. 
\item The algorithm's maximum capacity is large.
Specifically, in the known valuation case we need that the inventory cap for any item type $i\in[n]$ is 
$\frac{c}{\eps} \cdot \log(2vd)$ 
times the number of items of type $i$ in any bundle, for some large enough constant $c$. In the unknown valuation case, we need that this cap be at least 
$\frac{c}{\eps} \cdot \log(\frac{2dv}{\eps})$.
We note that the classical result from 30 years ago for the sell-only case \cite{AAP93} already required a similar assumption, so  obviously it is also necessary here. 
\item The algorithm is allowed free disposal, meaning that it may discard items from inventory at any point at no cost.
\end{enumerate}

In the unknown valuation setting, our incentive compatible mechanism is (almost)\footnote{The mechanism technically requires an extra bidding phase which we elaborate on soon.} a \emph{posted price} mechanism. The retailer posts prices \emph{per unit} for each item type, and these may change over time between different suppliers and customers. The price of each bundle $s\in S^t$, denoted by $p^t_s$, is the sum of prices of the items of $s$.\footnote{Technically, the prices of the bundles to customers is slightly more complex. See Section \ref{sec:techniques} for more details.}
Given these prices, each customer purchases a bundle maximizing her utility $v^t_s - p^t_s$ (so long as this utility is nonnegative) and is charged price $p^t_s$. Similarly, a supplier sells a bundle maximizing her utility $p^t_s - v^t_s$ paying a price $p^t_s$ (as long as this utility is nonnegative). 

We complement our algorithmic results with nearly matching lower bounds that show our competitive ratio is best possible up to constants. 
\begin{theorem}[Lower Bound, informal]\label{thm:lower-bound-informal}
For the online trading problem (even for the known valuation setting and single minded customers/suppliers), when comparing to an optimal offline fractional solution for which supplier values  are $(1+\eps)$ larger:
\begin{itemize}
\item The competitive ratio of any deterministic or randomized algorithm is $\Omega(\frac{1}{\eps}\cdot \log (dv))$. In particular, no algorithm can achieve a finite competitive ratio without resource augmentation (that is, when $\eps = 0$). 
   \item There exists a constant $c > 0$, such that if the inventory from an item type is smaller than $\frac{c}{\eps}\cdot \log(dv)$ times the number of items of type $i$ in at least one of the bundles, then the competitive ratio of any {\bf deterministic} algorithm is unbounded. 
\end{itemize}
\end{theorem}

\subsection{Techniques}\label{sec:techniques}

Our algorithms are natural to both describe and implement. Assume by scaling that the values of the customers for any bundle is in the range $[1,v]$, and that this range is known to the algorithm. At any time step $t\in [T]$, our algorithm maintains values $x_i^t$ for every item type $i\in [n]$ that can be viewed as a ``base" price 
per one unit of that item type. Our algorithm can be seen as a dynamic pricing algorithm that changes its prices per unit of item type based on the current inventory: when the inventory of an item type $i$ is full, the base price for that item is $0$, and the price increases exponentially as the inventory of the item type $i$ decreases. Dynamic pricing of this form is a common practice in retail and is used frequently (see e.g. \cite{den2015dynamic} for a survey). 

\paragraph{Known valuations:} The base price of any bundle $s\in S^t$, denoted by $p_s^t$, is the sum of base prices of the items in the bundle $s$.
Upon an arrival of a customer, the customer is allocated a bundle that maximizes $v_s^t-\max\{1,p_s^t\}$ if this value is non-negative, and is charged a price of $v_s^t$ for the bundle. 
We remark that in the known valuation setting we may choose a bundle that maximizes the standard utility of the customer defined as $v^t_s-p_s^t$. However, $\max\{1,p_s^t\}$ is used later in the unknown valuation setting in which we do not want to charge a customer with an arbitrarily small price, and would like the base price of a bundle to be at least $1$.  Upon arrival of a supplier, the price per unit of each item type is scaled down by a factor of $(1+\eps)$. The algorithm purchases from the supplier a bundle that maximizes $\frac{p_s^t}{1+\eps}- v_s^t$ if this value is non-negative, and pays a value of $v_s^t$ for the bundle.

Our analysis is done via a dual fitting approach that significantly generalizes previous dual fitting arguments that were used to analyze the customer-only setting \cite{BuchbinderN06,BuchbinderG15} (See also \cite{BN09} for a survey on the primal-dual approach for online algorithm). We first present a natural linear programming relaxation for the profit maximization objective; unlike the customer-only setting, this linear program is no longer a pure packing problem. Then we fit a dual to our algorithm's solution, which we use as an upper bound on the optimal profit. The dual has several moving parts; among other things it involves the prices $x_i^t$ generated during the algorithm execution. 

\paragraph{Unknown valuations:}
Obtaining an incentive compatible mechanism requires several technical steps.\footnote{We note that some of these ideas extend previous techniques that were used in \cite{AwerbuchAM03} to obtain an incentive compatible mechanism for the customer-only setting.} 
The algorithm starts by randomly sampling, once and for all before the online sequence begins, a value $\rho\in [0,v]$ from a carefully chosen distribution. Next, when a customer arrives the algorithm sets a price for bundle $s$ of $p'^{t}_s=\rho+\max\{1,p_s^t\}$. The customer purchases a bundle that maximizes her utility $v_s^t-p'^{t}_s$ as long as this utility is non-negative, and is charged a price of $p'^{t}_s$. A delicate technicality is that even if the customer decides \emph{not} to purchase the bundle because of this additional additive $\rho\geq 0$, but \emph{would have} purchase the bundle if $\rho$ was $0$, then the algorithm still updates the price per unit to future customers/suppliers as if the bundle was sold to the current customer (hence, our incentive compatible mechanism is not a simple posted price mechanism and requires a bidding phase). In the supplier case, the algorithm behaves almost identically to the known valuation setting, except that the algorithm posts a price of $\frac{p_s^t}{1+\eps} \geq v_s^t$ rather than the supplier's true value, $v_s^t$.

The algorithm for the known valuation setting has two free parameters that control the base prices $x_i^t$. We show that for the unknown valuation setting, if we carefully tune these parameters, the analysis in the known valuation setting extends naturally.

\subsection{Related Work}
\label{sec:related}

The most relevant related work to our setting is the following.

\paragraph{Customer only setting.} 
In this special case no suppliers arrive, and a given inventory should be allocated to arriving customers.
Awebrbuch, Azar and Plotkin~\cite{AAP93} initiated the study of the online routing problem, which is equivalent to a customer only version of our problem with known valuations.
Leonardi and Marchetti-Spaccamela \cite{LM99} later generalized their setting (see also \cite{BE98}, Chapter 13 for a textbook treatment). A series of work subsequently generalized this work yet further to capture combinatorial auctions; here the goal is to design incentive compatible mechanisms that maximize social welfare or revenue in both offline and online settings
\cite{AwerbuchAM03,BartalGN03,LaviS11, DobzinskiNS12, BuchbinderG15,FGL15}. 

We emphasize again that the ``large inventory'' assumption is required even for this easier customer-only setting~\cite{AAP93,LM99, BE98,BuchbinderG15}. However, our requirement is slightly higher and in particular depends on the value of $\epsilon$. Additionally, whereas the customer only settings \cite{BE98,BuchbinderG15} admits a competitive ratio that deteriorates (i.e. increases) smoothly as the size of the inventory shrinks, in our setting there is a threshold phenomenon. Our theorems show that for large enough $d,v$, there are constants $c_1<c_2$ such that if the inventory is larger than $\frac{c_2}{\eps} \log(2dv)$ times the number of items of type $i$ in any bundle, then our (deterministic) algorithm has logarithmic competitive ratio. However, if the inventory is smaller than $\frac{c_1}{\eps} \log(2dv)$ times the number of items of type $i$ in any bundle, then no deterministic algorithm can have a finite competitive ratio. We leave as an open question whether randomized algorithms can avoid this restriction on the inventory size, but this question has been open for 30 years even in the customer only setting.

\paragraph{Prophet trading.} Recently, Correa et al.~\cite{CDHOS23} initiated the study of a general {\em prophet trading} problem in which both buyers and sellers arrive sequentially. 
In their model (but our notation) the algorithm holds an inventory of at most $w$ items of a single item type, and faces a sequence of $T$ prices for this item. The prices are drawn from known distributions $F_1, \ldots, F_{T}$, and the realization of the prices are revealed in a random order. At each time step $t\in T$, the algorithm is allowed to both buy or sell items at the current price with the goal of maximizing the profit (the revenue obtained from selling the item minus the purchasing costs). They design a static single-threshold price algorithm (i.e. buy or sell depending on whether today's price is above or below a fixed threshold) that they show is constant competitive, and they also prove a constant lower bound. Very recent work~\cite{RCV25} extends their setting further to handle multiple item types and matroid constraints on the inventory. 

We may view \cite{CDHOS23} as a restricted stochastic case of our known valuation setting in which at each time step both a customer and a supplier arrive, and each wishes to buy/sell any quantity of a single item type at a fixed price per unit. The general problem we study allows for arbitrary bundles containing multiple item types, and adversarial customers/suppliers valuations for bundles. Because our setting is harder, we (a) only obtain logarithmic competitive ratios, and (b) need the additional assumptions of a large inventory, and of a weaker offline benchmark that sees supplier prices that are  $(1+\eps)$ higher. Our lower bounds show that both (a) and (b) are unavoidable.  Unlike  \cite{CDHOS23}, our algorithms require a more sophisticated dynamic pricing scheme.

Finally, we remark that \cite[Page 4]{CDHOS23} gives a simple example showing that no finite-competitive algorithm exists if prices are given in an adversarial order, even if that order is known to the algorithm beforehand. This bad example does not apply in our setting, because our model assumes that the inventory of the algorithm is initially full (or, alternatively, allows an additive constant in the competitive ratio).
We have preliminary results showing that if we \emph{do} assume full initial inventory, there is a simple $3$-competitive algorithm for the \emph{adversarial order} prophet trading problem, and $3$ is best possible. 

\bigskip 
Apart from the related work already discussed, there are several other important lines of work reminiscent of (but distinct!) from our problem.
\paragraph{(Online) Bilateral Trading.} Bilateral Trading has been studied extensively since the seminal work of Myerson and Satterthwaite~\cite{MYERSON1983265} (see also \cite{BlumrosenD21,BabaioffFN24}). 
Perhaps the closest version to our setting is the \emph{online} bilateral trade problem, where at every time step a buyer and seller arrive as a pair. Each has a private valuation functions for a good. The algorithm, which plays the role of the trading platform, posts a price for the good, and buyer and seller proceed with trade so long as both are willing to trade at this price. A common goal is to maximize the {\em gain from trade} defined as the \emph{sum} of utilities of the buyer and the seller~\cite{Cesa-BianchiCCF24, Global-Budget-regret,BachocCCC24}. This setting is very different from ours (even for a single item type): the objective is different, and furthermore buyers and sellers arrive and depart simultaneously, so the algorithm cannot stockpile goods in inventory for later trades. 

\paragraph{Two-Way Trading and Portfolio Selection.}
The problem of online portfolio selection has been extensively studied (See e.g., \cite{LH14} or \cite[Chapter 14]{BE98}).
For example, in a simple one way trading model introduced by El-Yaniv et al. \cite{EFKT01} a trader faces a a sequence of prices, and would like to maximize her profit from selling a single item. The closest to our setting is the Two-way trading problem in which a trader with an initial one unit of money observes a sequence of prices of a stock, and is allowed to buy/sell the stock at the given price with the goal of maximizing her final wealth~\cite{Fung19,Fung21}. This problem is, again, very different from our problem as there are no inventory constraints, and no bound on the demand/supply.  

\paragraph{Other.} 
Another nearby work is \cite{PW24}. Here goods appear and perish in inventory according to a Poisson process. Buyers also appear according to a Poisson process; these have linear or submodular valuation functions, and request bundles respecting downward closed constraint families (e.g. they only want want independent sets of a matroid). Similar problems have also been studied in the operations research community. See e.g, network revenue management problem (\cite[Chapter 7]{GT19}). This last line of research usually makes stochastic (as opposed to adversarial) assumptions about the input~\cite{GallegoR97, Jasin14, MaglarasM06, MaRST20,  abs-2403-05378}.

Finally, we remark that our linear formulation is an extension of the dual formulation of the {\em positive body chasing problem} introduced in \cite{BLS23}. A major difference is that in our problem we require an integral solution, and \cite{BLS23} does not maintain an integral dual. Furthermore we seek an incentive compatible mechanism, which is not a concern of \cite{BLS23}.

\section{Preliminaries}

\label{sec:prelim}

In this section, we formally define the problem considered in this paper and introduce notation.
\paragraph{Problem Statement}
The algorithm maintains an inventory of $n$ item types $i=1, \ldots, n$, and we use $[n]$ to denote the set $\{1,2, \ldots, n\}$. At any time $t$, the algorithm is required to hold an (integral) amount $r^t_i$ of item type $i$ such that $r^t_i \in [0, w_i]$ for some positive integer inventory $w_i\in \Z_{+}$. We assume that initially the inventory is full, $r_i^0 = w_i$ for all $i\in [n]$. Time steps $t\in [T]$ are partitioned into time steps $t\in \tbuy$ in which a customer arrives, and time steps $t\in \tsell$ in which a supplier arrives (i.e. $[T]=\tbuy\sqcup \tsell$). 

\begin{itemize}
    \item At time steps $t\in \tbuy$, a customer arrives and provides a menu of bundles $S^t$. Each bundle $s \in S^t$ contains $a_{s,i}\in \Z_+$ items of type $i$ and has a value $v^t_{s}$ to the customer. The customer would like to purchase at most a single bundle $s\in S^t$. 
    \item At time steps $t\in \tsell$, a supplier arrives and similarly provides a menu of bundles $S^t$. A bundle $s\in S^t$ has $a_{s,i}\in \Z_+$ items of type $i$ and has a value of $v^t_{s}$ to the supplier. The supplier would like to sell at most one bundle $s\in S^t$.
\end{itemize}

At any time step $t$, the algorithm can buy/sell up to a single bundle $s\in S^t$ from a supplier/customer. 
The (integral) inventory of items of type $i$, $r^t_i$, must remain in the range $[0, w_i]$ at all time steps. However, when the algorithm purchases a bundle from a supplier it may dispose any items that are not required for free if it already holds an inventory of $w_i$ items of type $i$ and cannot increase the inventory of this item type. 
We use $\tbuy'\subseteq \tbuy$ and $\tsell'\subseteq \tsell$ to refer to time steps in which the algorithm sells or buys a bundle respectively (as opposed to deciding \emph{not} to buy/sell), and $s^t_*$ is the bundle allocated to the customer, or bought from the supplier.
We study two different settings for this model: \begin{itemize}
    \item In the {\em known valuation setting} the values of the bundles are known to the algorithm and the price it charges or pays is the bundle's value. Hence, the profit of the algorithm is $v_{alg} = \sum_{t\in \tbuy'}v^t_{s^t_*} - \sum_{t\in \tsell'}v^t_{s^t_*}$.
    \item In the {\em unknown valuation setting} the algorithm we design an incentive compatible mechanism. Here, at time step $t\in \tbuy$, the algorithm posts (compactly) a price $p^t_s$ for each bundle $s\in S^t$ and the customer buys a bundle $s\in S^t$ maximizing her utility $v_s^t-p_s^t$ if this utility is non-negative.
    Similarly, at time step $t\in \tsell$, the algorithm posts prices $p^t_s$ for the bundles $s\in S^t$ and the supplier sells a bundle $s\in S^t$ that maximizes her utility $p_s^t-v_s^t$ if this utility is non-negative.
    The goal of the algorithm is maximizing its profit $v_{alg} = \sum_{t\in \tbuy'}p^t_{s^t_*} - \sum_{t\in \tsell'}p^t_{s^t_*}$.    
    \end{itemize}  

A special case of our setting is when the customers (or suppliers) are {\em single minded}, meaning that each customer would like to purchase a single bundle of items $s$ (or a bundle containing $s$) and has a value $v^t_s$ to this bundle. Similarly, a supplier would like to sell to the algorithm a single bundle $s$ (or a subset of the bundle $s$) and has a value $v^t_s$ to this bundle.

In the simpler setting of single minded customers/suppliers the bundle that maximizes the utility can be computed trivially. In case that the number of bundles in the menu is exponential, we assume that we are given a {\em demand oracle} to the customers/suppliers. That is, given a price per unit of each item type $x_i^{t-1}$ at the arrival of a customer, the customer is able to output a bundle maximizing its utility $v_s^t- \sum_{i=1}^{n}a_{i,s} \cdot x_i^{t-1}$. Similarly, upon an arrival of a supplier at time step $t$, it may output a bundle that maximizes $\sum_{i=1}^{n}a_{i,s} \cdot x_i^{t-1}- v_s^t$.  

Our algorithm is given a parameter $\eps\in(0, 1]$, and we compare the profit obtained by our algorithm with an optimal solution that maximizes the profit, but the value of each supplier to each bundle $s\in S^t$ is $(1+\eps) \cdot v_s^t$ instead of $v_s^t$. We note that instead of using a $(1+\eps)$ value augmentation for the suppliers, we could also use a $(1+\eps)$ reduction in the customers' values to achieve similar results, and we arbitrarily have chosen the first option.
As our proof is via duality of a fractional relaxation of the problem, the optimal solution can be fractional.

\paragraph{Additional Notation.}
 Define $d \triangleq \max_{t\in \tbuy, s^t\in S^t}\left\{\sum_{i=1}^n a_{s,i}\right\}$ be the largest size of a bundle of a customer, as well as $v_{\min} \triangleq \min_{t\in \tbuy, s\in S^t, v_s^t\neq 0}v^t_{s}$ and 
$v \triangleq \max_{t\in \tbuy, s\in S^t}v^t_{s}$ to be the minimum and maximum value of a bundle to a customer. We assume the values $d, v_{\min}, v$ are all known to the algorithm upfront. Without loss of generality we scale these values and assume that $v_{\min}=1$.

All logarithms in this paper are base $e$. We use a weighted generalization of KL divergence. Given
        a weight function $w$, define
\begin{equation}
    \wKL{x}{y} : = \sum_{i=1}^n w_i \left[x_i \log \left(\frac{x_i}{y_i}\right) -x_i + y_i\right]. \label{KL-inequality}
\end{equation}
It is known that $\wKL{x}{y} \geq 0$ for nonnegative vectors $x,y$ (one can check this is true term by term above).

\section{The Algorithm}\label{sec:algorithm}

In this section we prove our main theorem that we state here formally.
\begin{theorem}\label{thm:main-formal} 
For every instance of the online trading problem with a demand oracle for the valuations, and for every $\eps>0$, there is  
\begin{itemize}
    \item {\bf Known valuations:} 
    A deterministic online algorithm that is $O\left(\frac{1}{\eps} \log(2dv)\right)$-competitive provided that the inventory for any item type $i\in[n]$ is $\frac{c}{\eps} \cdot \log(2vd)$ times the number of items of type $i$ in any bundle for some large enough constant $c$. 
    \item {\bf Unknown valuations:} A randomized incentive compatible online algorithm that is $O\left(\frac{1}{\eps} \log(\frac{2dv}{\eps})\right)$-competitive in expectation provided that the inventory for any item type $i\in[n]$ is $\frac{c}{\eps} \cdot \log(\frac{2dv}{\eps})$ times the number of items of type $i$ in any bundle for some large enough constant $c$.
\end{itemize}
The competitive ratio of both algorithms is with respect to an optimal offline fractional solution, where supplier valuations at any time step 
are $(1+\eps)$ larger. 
\end{theorem}

In Section \ref{sec:known} we present our main algorithm for the known valuation setting proving the first part of the theorem. In Section \ref{sec:truthful} we show how to modify our algorithm in order to design an incentive compatible algorithm for the unknown valuations setting proving the second part of the theorem. 

\subsection{The Known Valuation Setting}\label{sec:known}

We first consider known valuation setting. Our algorithm has two parameter $\mu\geq 1$ and $\eta \geq 1+\log (1+vd\mu)$ that are chosen later. The algorithm requires the following assumption on the inventory size of each item type $i\in [n]$ compared with the number of items of type $i\in [n]$ in bundles presented to the algorithm.
\begin{assumption}[Large inventory]\label{assm1}
For any item type $i\in [n]$, time $t\in [T]$, and a bundle $s\in S^t$, we require that\vspace{-0.2cm}
\begin{equation*}
    w_i \geq \frac{8\eta}{\eps}\cdot a_{s,i}.
\end{equation*}
\end{assumption}
The formal description of our algorithm appears as Algorithm \ref{alg:allocate2}, and we first describe it less formally.
At any time step $t$, the algorithm maintains values $x_i^t$ for each item type $i$ that depends on the current inventory $r_i^t$. When the inventory of item $i$ is full $x_i^t= 0$, and $x_i^t$ increases exponentially as the inventory decreases to 0. 
Intuitively, at this point, the reader may think of the value $x_i^{t-1}$ as the base price per one unit of an item of type $i$ just before the arrival of the customer/supplier at time step $t$ (although our incentive compatible version of the algorithm in Section \ref{sec:truthful} requires a more delicate setting of the prices). Using this intuition, the price for a bundle $s\in S^t$ is $p_s^t=\sum_{i=1}^{n}a_{s,i}\cdot x_i^{t-1}$. Next,
\begin{itemize}
    \item {\bf Customer arrival:} As $v^t_s\geq 1$, the algorithm allocates a bundle that maximizes the utility of the customer $v^t_s-\max\{1,p_s^t\}$ if this utility is non-negative and charge the customer $v_s^t$. Otherwise, no bundle is allocated. 
    \item {\bf Supplier arrival:} The algorithm offers to the supplier prices that are $1+\eps$ times smaller, and buys a bundle that maximizes the utility of the supplier $\frac{p^t_s}{1+\eps}-v^t_s$ if it is non-negative, and otherwise no bundle is bought.
\end{itemize}

\begin{algorithm}[h]
\caption{Trade $(v,d, \eps)$}\label{alg:allocate2}
Let $r_i^{0}=w_i$ be the inventory of items of item type $i \in [n]$ at time step $0$.\\
Let $\mu\geq 1$, and $\eta\geq 1 + \log(1+vd\mu)$ be parameters chosen later.\\
At any time step $t \in [T]$ we set,
\begin{equation}
    \label{line:x_def}
   x^{t-1}_i \triangleq \frac{1}{d \cdot \mu}\left(\exp\left(\left(1-\frac{r_i^{t-1}}{w_i}\right)\cdot \eta \right) - 1\right).
\end{equation}\\ 
Upon arrival of a {\bf customer} at time step $t\in \tbuy$: \vspace{0.2cm}\\
\quad  Let $s^t_{*}=\arg\max_{s\in S^t}\left\{ \vt^t_{s}-\max\{1,\sum_{i=1}^{n}a_{s,i}\cdot x_i^{t-1}\}\right\}$. \label{step:chooses-sell}\\
\quad \lIf{$v^t_{s^t_{*}}-\max\{1,\sum_{i=1}^{n}a_{s^t_{*},i}\cdot x_i^{t-1}\} \geq 0$
}{Sell bundle $s^t_{*}$ to the customer and update the inventory to be
$r_{i}^t = r_{i}^{t-1} - a_{s^t_{*},i}$ \label{step:sell}
} \vspace{0.2cm}
Upon arrival of a {\bf supplier} at time step $t\in \tsell$: \nonumber \vspace{0.2cm}\\
\quad  Let $s^t_{*}=\arg\max_{s\in S^t}\left\{ \frac{1}{1+\eps}\cdot \sum_{i=1}^{n}a_{s,i} \cdot x_i^{t-1}-  \vt^t_{s}\right\}$. \label{step:chooses-buy}\\
\quad \lIf{$\frac{1}{1+\eps}\cdot \sum_{i=1}^{n}a_{s^t_{*},i}\cdot x_i^{t-1}-v^t_{s^t_{*}} \geq 0$}{Buy bundle $s^t_{*}$ from the supplier and update the inventory to be
$r_{i}^t = \min\{r_{i}^{t-1} + a_{s^t_{*},i},w_i\}$\label{line:inventory} 
}\label{step:buy}
\end{algorithm}

Let $\tbuy'\subseteq \tbuy$ and $\tsell'\subseteq \tsell$ be the time steps in which the algorithm sells or buys a bundle respectively (as opposed to deciding \emph{not} to buy/sell), and let $s^t_*$ be the bundle that was sold/bought from the customer/supplier.
In order to make our subsequent presentation simpler, define for every $t \in \tbuy' \cup \tsell'$ the quantity 
\begin{equation*}
  P^t \triangleq \sum_{i=1}^{n}a_{s^t_{*},i} \cdot x_{i}^{t-1},  
\end{equation*}
We prove the following theorem.
\begin{theorem}\label{thm:main-inside} 
Given parameter $\mu \geq 1$ and  $\eta\geq 1+\log (1+vd\mu)$, and assuming that for any item type $i\in [n]$, time $t\in [T]$, and a bundle $s\in S^t$, $w_i \geq \frac{8\eta}{\eps}\cdot a_{s,i}$,
Algorithm \ref{alg:allocate2} maintains a feasible integral inventory, such that
\[OPT = O\left(\frac{\eta}{\eps}\right) \left[  \sum_{t\in \tbuy'}\left(\left(1-\frac{\eps}{4}\right) P^t +\frac{\eps \vt^t_{s_{*}^t}}{\eta}+ \frac{1}{\mu}\right) - \sum_{t\in \tsell'}\frac{P^t}{1+\eps}\right] = O\left(\frac{\eta}{\eps}\right)  \left[\sum_{t\in \tbuy'} \vt^t_{s_{*}^t} - \sum_{t\in \tsell'}\vt^t_{s_{*}^t}\right].\]
where $OPT$ is an optimal offline fractional solution whose buying costs from a supplier at any time step $t\in \tsell, s\in S^t$ is $(1+\eps)\vt^t_{s}$. 
\end{theorem}
The first part of Theorem \ref{thm:main-formal} follows by setting $\mu=1$, and $\eta=1+\log(1+vd\mu)= 1+\log(1+vd)$.

We start with a couple of simple observations.
\begin{observation}
\label{obs:pos_prices}
    Item prices are always positive, i.e. for all $t \in [T]$ and $i \in [n]$, we have $x_i^t \geq 0$.
\end{observation}

\begin{proof}
    The algorithm explicitly maintains that $r_i^t \leq w_i$ in \cref{line:inventory}, which in turn implies that $x^{t-1}_i \triangleq \frac{1}{d \cdot \mu }\left(\exp\left((1-\frac{r_i^{t-1}}{w_i})\cdot \eta \right) - 1\right) \geq 0$.
\end{proof}

We make another observation that comes from rearranging the price update rule \eqref{line:x_def}.  
\begin{observation}
\label{obs:update rule}
Define $\widehat x_{i}^{t} \triangleq x_i^t+\frac{1}{d \cdot \mu}$. Then, for all $t \in [T]$ and $i \in [n]$, it holds that
    \[\frac{1}{\eta} \cdot \log\left(\frac{\widehat x_{i}^{t}}{\widehat x_{i}^{t-1}}\right) = \frac{r_i^{t-1} -r_i^t}{w_i}.\]
\end{observation}

\begin{proof}
    Manipulating the price update rule \eqref{line:x_def}, we get 
    $\log(\widehat x_{i}^{t-1})= \log(x_i^{t-1}+\frac{1}{d \cdot \mu}) = \left(1-\frac{r_i^{t-1}}{w_i}\right)\eta + \log(1/d\mu)$. Subtracting $\log(\widehat x_{i}^{t-1})$ from $\log(\widehat x_{i}^{t})$ and dividing by $\eta$ yields the claim. 
\end{proof}

Before bounding the competitive ratio, one might worry whether \cref{alg:allocate2} even maintains a feasible integral solution (namely why $r_i^t$ never goes below $0$). We prove that this is indeed the case.

\begin{lemma}
\cref{alg:allocate2} produces a feasible integral solution.
\end{lemma}

\begin{proof}
    Clearly, $r_i^t$ is integral since all bundles are integral. Furthermore, as already noted, we maintain $r_i^t \leq w_i$ explicitly. The only thing remaining to argue is that $r_i^t \geq 0$. 

    We claim that whenever the algorithm decides to sell a bundle $s^t_*$ with $a_{s^t_*,i} \geq 1$ to a customer, we have that $x_i^{t-1}\leq v \triangleq \max_{t\in \tbuy, s\in S^t}\vt^t_{s}$. To see this note that otherwise since $a_{s^t_*,i}\geq 1$ and by \cref{obs:pos_prices} $x_i^{t-1}\geq 0$ we have, $\max\{1,\sum_{i=1}^{n}a_{s^t_*,i}\cdot x_i^{t-1}\}\geq a_{s^t_*,i} \cdot x_i^{t-1} > v\geq \vt_{s_*^t}$, which means that the algorithm does not sell the bundle to the customer (see Line \ref{step:sell}).

    Rearranging \eqref{line:x_def}, whenever the algorithm decides to sell a bundle $s^t_*$ with $a_{s^t_*,i} \geq 1$:   
\begin{align*}
r_i^{t-1} & \geq w_i \left(1- \frac{\log\left(1+vd\mu\right)}{\eta}\right) \geq w_i \left(\frac{1 + \log(1+vd\mu) - \log\left(1 + vd\mu\right)}{\eta} \right) 
 \geq w_i \cdot \frac{\eps}{8\eta},   
\end{align*}
where the second equality follows since 
$\eta \geq 1 + \log(1+vd\mu)$, and the last inequality follows since $\eps\leq 1$. By \cref{assm1} we have that 
$a_{s^t_*,i} \leq w_i \cdot \frac{\eps}{8\eta}$, and we conclude that $r_i^t = r_i^{t-1}-a_{s^t_*,i} \geq 0$. We remark that the analysis here do not require the full strength of \cref{assm1} that is used later in the proof.
\end{proof}

\paragraph{Analysis via Duality:}
To prove the competitiveness stated in Theorem \ref{thm:main-inside} we present an LP formulation $\mathcal{P}$ for the fractional version of the problem in which supplier valuations are inflated by $(1+\eps)$. We then construct a feasible solution to the dual problem $\mathcal{D}$ whose value is $O(\frac{c+\log \mu}{\eps})$ times the value obtained by the algorithm. Theorem \ref{thm:main-inside} then follows directly by weak duality.

In the LP below, we may think of $\primy^t_{s}$ and $\primz_{s}^t$ respectively as the indicators for whether bundle $s\in S^t$ is allocated to the customer and supplier at time $t$, and as before, $\primr_i^t$ as the number of items of type $i$ in inventory at time $t$. The constraints are straightforward updating the inventory (which is in $[0,w_i]$), and requiring that at most a single bundle is allocated at any time $t$. Note that by making \eqref{cons1} and \eqref{cons2} inequalities (as opposed to equalities) we are allowing a free disposal of items. 

\[(\mathcal{P}): \max\sum_{t \in \tbuy}\sum_{s\in S^t}\vt^t_{s} \cdot \primy^{t}_{s}-(1+\eps) \cdot \sum_{t \in \tsell}\sum_{s\in S^t}\vt^t_{s} \cdot \primz^{t}_{s}\]\vspace{-0.2cm}
\begin{align}
\primr_{i}^{t+1}& \leq \primr_{i}^{t} - \sum_{s\in S^t}a_{s,i} \cdot \primy^{t}_{s} & 	\forall i \in [n], \ t\in \tbuy,   \label{cons1}\\
\primr_{i}^{t+1} & \leq \primr_{i}^{t} +  \sum_{s\in S^t}a_{s, i} \cdot \primz^{t}_{s} & 	\forall i \in [n], \ t\in \tsell, \label{cons2}\\
\sum_{s\in S^t} \primy^{t}_{s}& \leq 1 &  \forall t\in \tbuy, \label{cons3}\\
\sum_{s\in S^t} \primz^{t}_{s}& \leq 1 &  \forall t\in \tsell, \label{cons4}\\
\primr_i^t & \leq w_i & \forall t\in [T], \label{cons5}\\
\primy_{s}^t, \primz_{s}^t, \primr_i^t & \geq 0 &  
	\forall t\in [T], s\in S^t. \nonumber
\end{align}

The dual formulation has variables $x_i^t$ that correspond to Constraints \eqref{cons1} and \eqref{cons2}, variables $\alpha^t$ and $\beta^t$ corresponding to constraints \eqref{cons3} and \eqref{cons4}, and variables $\ell_{i}^{t}$ that correspond to Constraints \eqref{cons5}.

 \[(\mathcal{D}): \min\sum_{t=1}^{T}\sum_{i=1}^{n}w_i \cdot \duall_{i}^{t} +\sum_{t\in \tbuy} \duala^t + \sum_{t\in \tsell} \dualb^t\]\vspace{-0.5cm}
\begin{align}
\sum_{i=1}^{n}a_{s,i}\cdot \dualx_i^t+ \duala^t & \geq \vt^t_{s} & \forall t \in \tbuy, s\in S^t, \label{dual:cons1}\\
\sum_{i=1}^{n}a_{s,i} \cdot \dualx_i^t - \dualb^t & \leq \vt^t_{s} \cdot (1+\eps) & \forall t \in \tsell, s\in S^t, \label{dual:cons2}
\\
\duall_{i}^{t} & \geq \dualx_i^t-\dualx_{i}^{t-1} & \forall i \in [n], \ t\in [T], \label{dual:cons3}\\
\dualx_i^t, \duall_{i}^{t}, \duala^t, \dualb^t	 & \geq 0 & \forall i\in[n], \ t\in [T]. \nonumber
\end{align}

We have suggestively reused the name $x$ for the dual variables: indeed, we will soon use the values $x_i^t$ from \eqref{line:x_def} to set these. 

\paragraph{Constructing the dual solution.} 

We fit the following dual to Algorithm \ref{alg:allocate2}. Recall that $\tbuy'\subseteq \tbuy$ and $\tsell'\subseteq \tsell$ are the time steps in which the decides to sell or buy a bundle $s^t_*$ respectively. We set
\begin{align*}
    \dualx_i^t &= x_i^t, \\
    \duall_i^t &= \max\{0, x_i^t-x_i^{t-1}\}, \\
    \duala^t &= \begin{cases}
        \vt^t_{s^t_{*}} & \text{if } t \in \tbuy' \\
        0 & \text{otherwise}
    \end{cases}, \\
    \dualb^t &= \begin{cases}
        \sum_{i=1}^{n}a_{s^t_{*},i} \cdot x_i^{t-1}-(1+\eps)\vt_{s^t_{*}} & \text{if } t \in \tsell' \\
        0 & \text{otherwise}
        \end{cases}.
\end{align*}

We need to show dual feasibility, and that the cost of the algorithm is related to the dual cost. We start with the first.

\begin{lemma}\label{lem:main2}
The solution $(\dualx, \duall, \duala, \dualb)$ is feasible to $\mathcal{D}$ .
\end{lemma}

\begin{proof}
By construction we have $\duala^t\geq 0$, $\duall_i^t\geq 0$, and $\duall_i^t\geq x_i^t-x_i^{t-1}$, and by \cref{obs:pos_prices} also $x^{t-1}_i \geq 0$.

By the behavior of the algorithm (Line \ref{step:buy}), if $t\in \tsell'$, then $ \sum_{i=1}^{n}a_{s^t_{*},i} \cdot x_i^{t-1}-(1+\eps)\vt_{s^t_{*}}\geq 0$ and hence $\dualb^t\geq 0$. Hence, it remains to check constraints \eqref{dual:cons1} and \eqref{dual:cons2}.

Consider first any time $t\in \tbuy \setminus \tbuy'$ in which no bundle is sold. 
In this case, by line \ref{step:sell} of \cref{alg:allocate2}, for all $s\in S^t$ we have $\max\{1,\sum_{i=1}^{n}a_{s,i} \cdot x_i^{t-1}\} > \vt^t_{s}$ (note the strict inequality), as $\vt^t_{s}\geq 1$ for any bundle $s\in S^t$, it means that for every $s\in S^t$, $\sum_{i=1}^{n}a_{s,i} \cdot x_i^{t-1}> \vt^t_{s}$ and therefore setting $\duala^t=0$  and noticing that in this case $x_i^{t}= x_i^{t-1}$ satisfies constraints \eqref{dual:cons1}. Next, consider any time $t\in \tbuy'$ in which a bundle $s^t_{*}=\arg \max_{s\in S^t}\left\{ \vt^t_{s}-\max\{1,\sum_{i=1}^{n}a_{s,i}\cdot x_i^{t-1}\}\right\}$ is allocated to a customer. By the definition of $s^t_*$, we have for any $s\in S^t$
\begin{align*}
 \duala^t & = v^t_{s^t_{*}} \geq v^t_{s} +\max\{1,\sum_{i=1}^{n}a_{s^t_{*},i}\cdot x_i^{t-1}\} - \max\{1,\sum_{i=1}^{n}a_{s,i}\cdot x_i^{t-1}\} \\ & \geq  \vt^t_{s} +1 -\max\{1,\sum_{i=1}^{n}a_{s,i}\cdot x_i^{t-1}\}  \geq 
 \vt^t_{s}-\sum_{i=1}^{n}a_{s,i}\cdot x_i^{t-1}.   
\end{align*}
Thus, we get that for any $s\in S^t$,  $\sum_{i=1}^{n}a_{s,i}\cdot x_i^t +\duala^t \geq  \sum_{i=1}^{n}a_{s,i}\cdot x_i^{t-1} +\duala^t \geq \vt^t_{s}$, where 
the first inequality holds since the values $x_i^t$ only increase at time steps $t\in \tbuy$.

Similarly, let $t \in \tsell \setminus \tsell'$ be a time in which no bundle is allocated to the supplier. \cref{step:buy} of \cref{alg:allocate2} guarantees that in this case for all $s\in S^t$ we have
$\sum_{i=1}^{n}a_{s,i}\cdot x_i^{t-1} < (1+\eps) \cdot \vt_{s}$, and therefore setting $\dualb^t=0$ satisfies constraints \eqref{dual:cons2}. Finally, at times $t \in \tsell'$, in which a bundle $s^t_{*}=\arg\max_{s\in S^t}\left\{ \frac{1}{1+\eps}\cdot \sum_{i=1}^{n}a_{s,i} \cdot x_i^{t-1}-  \vt^t_{s}\right\}$ is bought from the supplier, we have for all $s \in S^t$,
\[\sum_{i=1}^{n}a_{s,i}\cdot x_i^t -\dualb^t \leq \sum_{i=1}^{n}a_{s,i}\cdot x_i^{t-1}- \max_{s\in S^t}\left\{\sum_{i=1}^{n}a_{s,i}\cdot x_i^{t-1}- (1+\eps)\cdot \vt^t_{s}\right\} \leq (1+\eps)\cdot \vt^t_{s}.\]
The first inequality holds since the values $x_i^t$ only decrease at time steps $t\in \tsell$.
Hence, we satisfy the dual constraints \eqref{dual:cons1} and \eqref{dual:cons2}.
\end{proof}

The remaining challenge is to relate the cost of the algorithm to the cost of the dual solution. This is proved in \cref{lem:main3}. Theorem \ref{thm:main-inside} follows directly by combining \cref{lem:main2} and \cref{lem:main3} along with weak duality. 
\begin{lemma}\label{lem:main3}
Let $P^t \triangleq \sum_{i=1}^{n}a_{s^t_{*},i} \cdot  x_{i}^{t-1}$ for every time step $t\in \tsell' \cup \tbuy'$. Then, the value of the dual solution is bounded as
\begin{align*}
    \sum_{\substack{t \in [T] \\ i \in [n]}} w_i \cdot \duall_{i}^{t} +\sum_{t\in \tbuy} \duala^t + \sum_{t\in \tsell} \dualb^t &= O\left(\frac{\eta}{\eps}\right) \cdot \left[  \sum_{t\in \tbuy'}\left(\left(1-\frac{\eps}{4}\right)P^t +\frac{\eps \cdot \vt^t_{s_{*}^t}}{\eta}+ \frac{1}{\mu}\right) - \sum_{t\in \tsell'}\frac{P^t}{1+\eps}\right]\\
    & =O\left(\frac{\eta}{\eps}\right) \cdot \left[\sum_{t\in \tbuy'} \vt^t_{s_{*}^t} - \sum_{t\in \tsell'}\vt^t_{s_{*}^t}\right].
\end{align*}
\end{lemma}
Proving \cref{lem:main3} requires several  intermediate claims. 

\begin{claim}
\label{claim:kl_bounds}
Recall that $\widehat x_{i}^{t} \triangleq x_i^t+\frac{1}{d \cdot \mu}$.
For time steps in which the trader allocates a bundle to a customer or a supplier we have the following.
    \begin{align}
\forall t\in \tbuy' \qquad & \frac{1}{\eta}\sum_{i=1}^{n}w_i \cdot \widehat x_{i}^{t} \log\left(\frac{\widehat x_{i}^{t}}{\widehat x_{i}^{t-1}}\right)  \leq e^{\eps/8} \cdot\left(P^t + \frac{1}{\mu}\right),   \label{ineq-sell}\\
\forall t\in \tsell' \qquad &  \frac{1}{\eta}\left[\sum_{i \mid x_i^t>0} w_i \cdot \widehat x_{i}^{t} \log\left(\frac{\widehat x_{i}^{t}}{\widehat x_{i}^{t-1}}\right) - \sum_{i \mid x_i^t=0}w_i \cdot x_i^{t-1}\right]   \leq  - e^{\eps/8}\cdot P^t.  \label{ineq-buy}
\end{align}
\end{claim}
\begin{proof}
Consider a time step $t\in \tbuy'$. Then,
\begin{align}
\frac{1}{\eta}\sum_{i=1}^{n}w_i \cdot \widehat x_{i}^{t} \log\left(\frac{\widehat x_{i}^{t}}{\widehat x_{i}^{t-1}}\right) 
& = \sum_{i=1}^{n} a_{s^t_{*},i} \cdot \widehat x_{i}^{t}  =  \sum_{i=1}^{n}a_{s^t_{*},i} \cdot  \widehat x_{i}^{t-1}\cdot \exp\left(\frac{a_{s^t_{*},i}}{w_i}\cdot \eta \right) \label{eq11}\\
& \leq e^{\eps/8} \cdot \sum_{i=1}^{n} a_{s^t_{*},i} \cdot \widehat x_{i}^{t-1} = e^{\eps/8} \cdot \sum_{i=1}^{n}a_{s^t_{*},i} \cdot \left(x_{i}^{t-1} + \frac{1}{d \cdot \mu}\right) \label{ineq12}\\
& \leq  e^{\eps/8} \cdot\left(P^t + \frac{1}{\mu}\right). \label{eq13}
\end{align}
Step \eqref{eq11} follows by using \cref{obs:update rule} twice, and the fact that $r_i^{t-1} -r_i^t = a_{s^t_*,i}$. 
Inequality \eqref{ineq12} follows by \cref{assm1}.
Inequality \eqref{eq13} follows since $d\geq \sum_{i=1}^{n}a_{s^t_{*},i}$ 
Next, at a time  step $t\in \tsell'$  we have,

\begin{align}
&\frac{1}{\eta}\sum_{i \mid x_i^t>0} w_i \cdot \widehat x_{i}^{t} \log\left(\frac{\widehat x_{i}^{t}}{\widehat x_{i}^{t-1}}\right) -\frac{1}{\eta} \sum_{i \mid x_i^t=0}w_i \cdot x_i^{t-1} \leq \frac{1}{\eta}\sum_{i \mid x_i^t>0} w_i \cdot \widehat x_{i}^{t} \log\left(\frac{\widehat x_{i}^{t}}{\widehat x_{i}^{t-1}}\right) -\sum_{i \mid x_i^t=0}a_{s^t_{*},i}\cdot x_i^{t-1} \label{eq52}\\
& = -\sum_{i \mid x_i^t>0} a_{s^t_{*},i}\cdot  \widehat x_{i}^{t}  -\sum_{i \mid x_i^t=0}a_{s^t_{*},i} \cdot x_i^{t-1}  
 = -\sum_{i \mid x_i^t>0}a_{s^t_{*},i} \cdot \widehat x_{i}^{t-1}\cdot \exp\left(-\frac{a_{s^t_{*},i}}{w_i}\cdot \eta\right) -\sum_{i \mid x_i^t=0}a_{s^t_{*},i}\cdot x_i^{t-1} \label{eq21} \\
& \leq -e^{-\eps/8}\sum_{i \mid x_i^t>0} a_{s^t_{*},i}\cdot \widehat x_{i}^{t-1} -\sum_{i \mid x_i^t=0}a_{s^t_{*},i}\cdot x_i^{t-1}  \label{ineq22} \leq - e^{-\eps/8}\sum_{i=1}^{n} a_{s^t_{*},i}\cdot x_{i}^{t-1}  = - e^{-\eps/8}\cdot P^t. 
\end{align}
Inequality \eqref{eq52} follows by \cref{assm1} and the fact that $\epsilon \leq 1$. 
Equality \eqref{eq21} follow by using \cref{obs:update rule} twice, and by the fact that if $x_i^t>0$ then $r_i^t<w_i$, which means $r_i^{t-1} -r_i^t = -a_{s^t_*,i}$.
The first inequality of line \eqref{ineq22} follows from \cref{assm1}. The final inequality follows since $\widehat x_{i}^{t-1}\geq x_{i}^{t-1}$.
\end{proof}

\begin{claim}
\label{lemmaKL}
The following inequality holds: $\sum_{t \in \tbuy'} \left(P^t + \frac{1}{\mu}\right)  - e^{-\eps/4} \sum_{t \in \tsell'} P^t
\geq 0$.
\end{claim}
\begin{proof}
By the non-negativity of the KL divergence, we have that for all $t \in \tbuy' \cup \tsell'$,
\begin{align*}
 0 \leq \frac{1}{\eta}\sum_{i \mid x_i^t>0}w_i \cdot \left[\widehat x_i^t\log\left(\frac{\widehat x_{i}^{t}}{\widehat x_{i}^{t-1}}\right) -\widehat x_i^t + \widehat x_i^{t-1}\right] = \frac{1}{\eta}\sum_{i \mid x_i^t>0}w_{i}\cdot \left[\widehat x_{i}^{t}\log\left(\frac{\widehat x_{i}^{t}}{\widehat x_{i}^{t-1}}\right) - x_i^t +  x_i^{t-1}\right].
\end{align*}
Summing up the inequalities for all time steps $t\in \tbuy'\cup \tsell'$, we get 
\begin{align}
	0  \leq & \frac{1}{\eta} \sum_{t \in \tbuy'\cup \tsell'} \sum_{i\mid x_i^t>0} w_i \cdot \left[\widehat x_{i}^{t} \cdot\log\left(\frac{\widehat x_{i}^{t}}{\widehat x_{i}^{t-1}}\right) - x_i^t + x_i^{t-1}\right] + \frac{1}{\eta}\sum_{t\in \tsell'}\sum_{i \mid x_i^t=0}w_i \cdot \left[-x_{i}^{t-1} - x_i^t +  x_i^{t-1}\right] \label{ineq:telescope-0} \\
	=  & \frac{1}{\eta} \sum_{t \in \tbuy'} \sum_{i=1}^{n} w_i \cdot \widehat x_{i}^{t} \cdot \log\left(\frac{\widehat x_{i}^{t}}{\widehat x_{i}^{t-1}}\right) + \frac{1}{\eta} \sum_{t \in \tsell'} \left[\sum_{i \mid x_i^t>0}  w_i\cdot \widehat x_i^t\log\left(\frac{\widehat x_{i}^{t}}{\widehat x_{i}^{t-1}}\right)- \sum_{i \mid x_i^t=0}x_i^{t-1}\right] \nonumber \\
          & - \frac{1}{\eta} \sum_{i=1}^n w_i\cdot (x^T_i - x^0_i)
 \label{ineq:telescope} \\
 \leq & e^{\eps/8} \cdot \sum_{t \in \tbuy'} \left(P^t + \frac{1}{\mu}\right) - e^{-\eps/8} \cdot \sum_{t \in \tsell'} P^t. \label{ineq:yz}
\end{align}
Inequality \eqref{ineq:telescope-0} follows by adding a term that equals to 0.
Equality \eqref{ineq:telescope} is a telescoping sum.
Finally, inequality \eqref{ineq:yz} follows by plugging inequalities \eqref{ineq-sell} and \eqref{ineq-buy} from \cref{claim:kl_bounds}, and since $x_i^T-x_i^0=x_i^T\geq 0$. Dividing by $e^{\eps/8}$ concludes the proof of the claim.
\end{proof}

We are finally ready to prove Lemma \ref{lem:main3}.
\begin{proof}[Proof of \cref{lem:main3}]

First, for all time steps $t\in \tbuy'$, we have \begin{align}\sum_{i=1}^n w_i \duall_i^t =  \sum_{i=1}^{n}w_i  (x_{i}^{t}-x_{i}^{t-1})  =  \sum_{i=1}^{n}w_i (\widehat x_{i}^{t}- \widehat x_{i}^{t-1}) \leq  \sum_{i=1}^{n}w_i \cdot \widehat x_{i}^{t} \log\left(\frac{\widehat x_{i}^{t}}{\widehat x_{i}^{t-1}}\right) \leq 2\eta \left(P^t + \frac{1}{\mu} \right). \label{line:movement}\end{align} 
The penultimate inequality follows because for any $a\geq b>0$, we have $(a-b)\leq a\log(\nicefrac{a}{b})$. The last inequality is due to \cref{claim:kl_bounds} and since $\eps\leq 1$.

Using this together with \cref{lemmaKL}, we bound the value of the dual solution (except for the value of $\duala^t$) as follows:
\begin{align}
&\sum_{\substack{t\in [T] \\ i \in [n]}} w_i\cdot \duall_{i}^{t} + \sum_{t\in \tsell} \dualb^t \leq  \sum_{t\in \tbuy'}\left[2\eta \cdot \left(P^t + \frac{1}{\mu}\right)\right] + \sum_{t\in \tsell'}\left(P^t- (1+\eps)\vt_{s_{*}^t}\right) \label{final-ineq2} \\
  & \leq \sum_{t\in \tbuy'}\left[2\eta \cdot \left(P^t + \frac{1}{\mu}\right)\right] + \sum_{t\in \tsell'} P^t + \frac{30\eta}{\eps}\cdot \left[\sum_{t \in \tbuy'} \left(P^t + \frac{1}{\mu}\right) - e^{-\eps/4}\cdot \sum_{t \in \tsell'} P^t\right] \label{final-ineq4}\\
  & = \sum_{t\in \tbuy'}\left[\left(2\eta+\frac{30\eta}{\eps}\right) \cdot \left(P^t + \frac{1}{\mu}\right)\right] - \sum_{t\in \tsell'}\left(\frac{30\eta}{\eps}\cdot(1+\eps)e^{-\eps/4} - (1+\eps)\right)\frac{P^t}{1+\eps} \nonumber \\
  &\leq \left(2\eta +\frac{30\eta}{\eps}\right) \cdot \sum_{t\in \tbuy'}\left(P^t + \frac{1}{\mu}\right) - \sum_{t\in \tsell'}\left(\frac{30\eta}{\eps} + 15\eta -2 \right)\frac{P^t}{1+\eps} \label{final-ineq5} \\
  &\leq \left(2\eta +\frac{30\eta}{\eps}\right) \cdot \sum_{t\in \tbuy'}\left(P^t +\frac{1}{\mu}\right) - \left(13\eta+\frac{30\eta}{\eps}\right) \cdot \sum_{t\in \tsell'}\frac{P^t}{1+\eps}  \label{final-ineq6} \\ 
  &\leq \left(13\eta +\frac{30\eta}{\eps}\right) \cdot \left(1-\frac{\eps}{4}\right) \cdot \sum_{t\in \tbuy'}\left(P^t +\frac{1}{\mu}\right) - \left(13\eta+ \frac{30\eta}{\eps}\right) \cdot \sum_{t\in \tsell'}\frac{P^t}{1+\eps}.  \label{final-ineq65}  
\end{align}

Step \eqref{final-ineq2} follows from \eqref{line:movement} and the fact that 
  $\duall^t$, $\dualb^t$ are all $0$ when $t \not\in \tbuy' \cup \tsell'$.
Step \eqref{final-ineq4} follows from adding $\frac{30\eta}{\eps}$ times the inequality of 
\cref{lemmaKL}.
Step \eqref{final-ineq5} holds since $\eps\leq 1$ and  $(1+\eps)e^{-\eps/4}\geq (1+\eps)(1-\eps/4)\geq (1+\eps/2)$, and \eqref{final-ineq6} since $\eta\geq 1$.
Step \eqref{final-ineq65} follows since 
\[2\eta +\frac{30\eta}{\eps}\leq \frac{30\eta}{\eps} + 13\eta - \frac{30}{4}\eta - \frac{13}{4}\eta \leq \frac{30\eta}{\eps} + 13\eta - \frac{30}{4}\eta - \frac{13\eps}{4}\eta=  \left(13\eta +\frac{30\eta}{\eps}\right) \cdot \left(1-\frac{\eps}{4}\right). \]

Finally, by the construction of the dual solution, we have
$\sum_{t\in \tbuy} \duala^t =  \sum_{t\in \tbuy'} \vt^t_{s_{*}^t}$.
Adding this to the final inequality, we get:
\begin{align*}
&\sum_{\substack{t\in [T] \\ i \in [n]}} w_i\cdot \duall_{i}^{t} + \sum_{t\in \tsell} \dualb^t + \sum_{t\in \tbuy} \duala^t \nonumber \\
&\leq \sum_{t\in \tbuy'} \vt^t_{s_{*}^t}+ \left(13\eta +\frac{30\eta}{\eps}\right) \cdot \left[\left(1-\frac{\eps}{4}\right) \cdot \sum_{t\in \tbuy'}\left(P^t +\frac{1}{\mu}\right) -  \sum_{t\in \tsell'}\frac{P^t}{1+\eps} \right] \nonumber \\   
& = O\left(\frac{\eta}{\eps}\right) \cdot \left[  \sum_{t\in \tbuy'}\left(\left(1-\frac{\eps}{4}\right)P^t +\frac{\eps \cdot \vt^t_{s_{*}^t}}{\eta}+ \frac{1}{\mu}\right) - \sum_{t\in \tsell'}\frac{P^t}{1+\eps}\right] = O\left(\frac{\eta}{\eps}\right) \cdot \left[\sum_{t\in \tbuy'} \vt^t_{s_{*}^t} - \sum_{t\in \tsell'}\vt^t_{s_{*}^t}\right]. \label{final-ineq7}
\end{align*}

The last inequality follows firstly since as $\eta\geq 1, \mu \geq 1$,  $\eps\leq 1$, and $\vt^t_{s^t_*} \geq 1$ for $t \in \tbuy$. For every $t\in \tbuy'$, we have $\eps \cdot \vt^t_{s_{*}^t} / \eta + 1/\mu \leq 2 \vt^t_{s_{*}^t}$.
More crucially, by the 
properties of the algorithm at each time step $t\in \tbuy'$ in which the algorithm sells a bundle to a customer, we have $P^t=\sum_{i=1}^{n}a_{s^t_*,i}\cdot x_i^{t-1} \leq \max\{1,\sum_{i=1}^{n}a_{s^t_*,i}\cdot x_i^{t-1}\} \leq \vt^t_{s^t_*}$ (see \cref{step:sell}), and at each time step $t\in \tsell'$ in which the algorithm buys a bundle from a supplier, we have $P^t=\sum_{i=1}^{n}a_{s^t_*,i}\cdot x_i^{t-1} \geq (1+\eps)\cdot \vt^t_{s^t_*}$ (see \cref{step:buy}).
\end{proof}

\subsection{The Unknown Valuation Setting}\label{sec:truthful}

In this section we prove the second part of Theorem \ref{thm:main-formal} by designing an incentive compatible algorithm.
Our incentive compatible algorithm is based on \cref{alg:allocate2} and its analysis with the following changes. We initially sample a random threshold $\rho\geq 0$ (in a way that is described formally later). Then,
\begin{itemize}
    \item When a customer arrives, the algorithm sets a price for each bundle $p^t_s \triangleq  \rho+\max\{1,\sum_{i=1}^{n}a_{s,i}\cdot x_i^{t-1}\}$ 
    for every bundle $s \in S^t$. The customer buys the bundle $s_*^t$ maximizing her utility $\vt^t_{s}-p^t_s$ if this utility is nonnegative. Additionally, \emph{regardless of whether any bundle is allocated in round $t$}, if $\vt^t_{s_{*}^t}-\max\{1,\sum_{i=1}^{n}a_{s_{*}^t,i}\cdot x_i^{t-1}\}$ is non-negative, the algorithm subtracts the bundle contents from the inventory and updates the values $x_i^t$ accordingly as done by our algorithm for the known value setting, \cref{alg:allocate2}. Note that as $\rho\geq 0$, then whenever $\vt^t_{s_{*}^t}- (\rho+\max\{1,\sum_{i=1}^{n}a_{s^t_*,i}\cdot x_i^{t-1}\})\geq 0$, then $\vt^t_{s_{*}^t}-\max\{1,\sum_{i=1}^{n}a_{s^t_*,i}\cdot x_i^{t-1}\}$ is also non-negative.
Formally, as the machanism may update the prices even if a bundle is not allocated to the customer, the incentive compatible mechanism is not a posted price mechanism, and requires a bidding phase.   
    \item When a supplier arrives, the algorithm sets a price $p^t_s \triangleq \frac{1}{1+\eps}\sum_{i=1}^{n}a_{s,i} \cdot x_i^{t-1}$ for every bundle $s\in S^t$. The supplier sells the bundle $s_*^t$ maximizing his utility $p^t_s- \vt^t_{s}$ if this utility is nonnegative.   
\end{itemize}

In addition to the above changes, the algorithm chooses the parameters $\mu$ and $\eta$ carefully as we describe in the formal description of \cref{alg:allocate-truthful}.

\begin{algorithm}[h]
\caption{Trade-Truthful $(v,d,\eps)$}\label{alg:allocate-truthful}
Let $r_i^{0}=w_i$ be the inventory of items of item type $i=1, \ldots, n$ at time step $0$.\\
Set the parameters $\mu= \frac{32}{\eps}\cdot (1+\log v)$, and $\eta= 32(1+\log(1+dv\mu))$. \label{step:parameters}\\
Let $\delta=\frac{\eps}{8}$. Choose randomly a value $\rho\geq 0$ as follows, 
\begin{equation}\rho = \begin{cases}
    0 & \text{with probability } 1-\delta \\
    2^j  & j\in \{0,1, \ldots, \lfloor\log v\rfloor \}, \text{with probability } \frac{\delta}{1+\lfloor\log v\rfloor}
\end{cases}, \label{equality-prob}\end{equation} \label{step:choosep2}\\
At any time step $t=1, \ldots, T$ we set,
\begin{equation}
   x^{t-1}_i \triangleq \frac{1}{d \cdot \mu}\left(\exp\left((1-\frac{r_i^{t-1}}{w_i})\cdot \eta \right) - 1\right).
\end{equation}\\
Upon arrival of a {\bf customer} at time step $t\in \tbuy$: \vspace{0.2cm}\\
\quad Set the price of bundle $s\in S^t$ to be $p_s^t=\rho+\max\{1,\sum_{i=1}^{n}a_{s,i}\cdot x_i^{t-1}\}$. \label{step:posted1}\\ 
\quad Let $s^t_{*}=\arg\max_{s\in S^t}\left\{\vt^t_{s}-p_s^t\right\}= \arg\max_{s\in S^t}\left\{ \vt^t_{s}-\max\{1,\sum_{i=1}^{n}a_{s,i}\cdot x_i^{t-1}\}\right\}$. \label{step:chooses-sell1}\\
\quad \lIf{$\vt^t_{s^t_{*}}-\max\{1,\sum_{i=1}^{n}a_{s^t_{*},i}\cdot x_i^{t-1}\} \geq 0$
}{update the inventory to be
$r_{i}^t = r_{i}^{t-1} - a_{s^t_{*},i}$ \label{step:update-inventory1}}\vspace{0.2cm}
\quad \lIf{$\vt^t_{s^t_{*}}-p^t_{s^t_{*}} =\vt^t_{s^t_{*}}-\max\{1,\sum_{i=1}^{n}a_{s^t_{*},i}\cdot x_i^{t-1}\}- \rho \geq 0$
}{ Sell bundle $s^t_{*}$ to the customer and charge a price of $p^t_{s^t_{*}}$} \label{step:sell1} 
\vspace{0.2cm}
Upon arrival of a {\bf supplier} at time step $t\in \tsell$: \nonumber \vspace{0.2cm}\\
\quad Set the price of bundle $s\in S^t$ to be $p_s^t=\frac{1}{1+\eps}\cdot \sum_{i=1}^{n}a_{s,i}\cdot x_i^{t-1}$. \label{step:posted2}\\ 
\quad  Let $s^t_{*}=\arg\max_{s\in S^t}\left\{p_s^t-\vt^t_{s}\right\}= \arg\max_{s\in S^t}\left\{ \frac{1}{1+\eps}\cdot \sum_{i=1}^{n}a_{s,i} \cdot x_i^{t-1}-  \vt^t_{s}\right\}$. \label{step:chooses-buy1}\\
\quad \lIf{$p^t_{s^t_{*}} -\vt^t_{s^t_{*}} \geq 0$}{Buy bundle $s^t_{*}$ from the supplier, pay a price of $p_{s^t_{*}}^t$, and update the inventory to be
$r_{i}^t = \min\{r_{i}^{t-1} + a_{s^t_{*},i},w_i\}$\label{line:inventory1}
}
\end{algorithm}

\begin{remark}
    We remark that in Step \ref{step:update-inventory1} the inventory is updated even if eventually in Step \ref{step:sell1} the algorithm does not allocate the bundle to the customer (since $\rho$ is too large). The algorithm can be lazy and delay disposing items until the inventory exceeds the capacity. However, for our analysis, we require that the algorithm updates the values $x_i^t$ \textbf{as if the items were allocated}. \end{remark}

We begin with the following observation.
\begin{observation}
    \cref{alg:allocate-truthful} is incentive compatible.
\end{observation}
\begin{proof}
    We observe that \cref{alg:allocate-truthful} sets prices for each bundle in Steps \ref{step:posted1} and \ref{step:posted2} allocates to the customers/suppliers a bundle $s^t_*$ that maximizes their utility with respect to these prices (if it is non-negative).
Hence, we get that the algorithm is incentive compatible.
\end{proof}
Next, we have the following claim.
\begin{claim}\label{claim:same}
At any time step $t\in [T]$, the inventory of the algorithm $r_i^t$, the values $x_i^t$ and the identity of $s^*_t$ is the same in \cref{alg:allocate-truthful} as in \cref{alg:allocate2}.
\end{claim}

\begin{proof}
    The claim follows inductively on the time steps $t\in T$. We observe that whenever a supplier arrives, the allocation is identical in both algorithms (only the price paid is different). Whenever a customer arrives at time step $t\in \tbuy$, as $\rho\geq 0$ is simply an additive shift, the bundle $s^t_{*}$ chosen in Step \ref{step:chooses-sell} in \cref{alg:allocate2} is the same as the bundle chosen by \cref{alg:allocate-truthful} in Step \ref{step:chooses-sell1}. Moreover, even if the bundle $s^t_{*}$ is not allocated to the customer in Step \ref{step:sell1} (which can happen if $\rho$ is too large), the algorithm still updates its inventory in Step \ref{step:update-inventory1} (as in Step \ref{step:sell} of \cref{alg:allocate2}) as well as the values $x_i^t$ that depend on the inventory.
\end{proof}

By the above claim, we get by the analysis of \cref{alg:allocate2} that the inventory of \cref{alg:allocate-truthful} is always feasible.
Moreover, the dual solution $\mathcal D$ as constructed in Section \ref{sec:known} is also feasible. Therefore, setting the parameters 
 $\mu= \nicefrac{32}{\eps}\cdot (1+\log v)$, and $\eta= 32(1+\log(1+dv\mu))$ (Line \ref{step:parameters} of \cref{alg:allocate-truthful}) we get,
by Theorem \ref{thm:main-inside} that, 
\[OPT = O\left(\frac{\eta}{\eps}\right) \cdot \left[\sum_{t\in \tbuy'}\left(\left(1-\frac{\eps}{4}\right)\cdot P^t +\frac{\eps \cdot \vt^t_{s_{*}^t}}{\eta}+ \frac{1}{\mu}\right) - \sum_{t\in \tsell'}\frac{P^t}{1+\eps}\right].\]

Next, let $\tbuy', \tsell'$ be the time steps in which \cref{alg:allocate2} sells bundle $s^t_{*}$ to the customer at price $v^t_{s^t_{*}}$ or pays the supplier a cost of $v^t_{s^t_{*}}$. \cref{alg:allocate-truthful}  updates its inventory the same way in these steps. However, it pays the supplier a higher price of $p_{s^t_{*}}^t\geq \vt^t_{s^t_{*}}$ and charges the customer a (random) lower price of $p_{s^t_{*}}^t\leq \vt^t_{s^t_{*}}$ if the bundle $s^t_{*}$ is allocated to the customer. Nevertheless, the next lemma bounds from below the expected profit of \cref{alg:allocate-truthful}.

\begin{lemma}\label{lem:revenue}
The expected profit of \cref{alg:allocate-truthful} is at least, 
 \[\left[\sum_{t\in \tbuy'}\left(\left(1-2\delta \right)\cdot \max\{1, P^t\} + \frac{\delta}{2(1+\log v)} \cdot v^t_{s^t_*}\right) - \sum_{t\in \tsell'}\frac{P^t}{1+\eps}\right].\]
\end{lemma}
\begin{proof}
By construction, for all time steps $t
\in \tsell'$, the algorithm pays $\frac{P^t}{1+\eps}$. Consider a time step $t\in \tbuy'$ in which $\max\{1,P^t\} = \max\{1,\sum_{i=1}^{n}a_{s^t_*,i}\cdot x_i^{t-1}\}\leq \vt^t_{s^t_*}$. 
If $\max\{1, P^t\} \leq \vt^t_{s^t_*} \leq \max\{1, P^t\}+ 2^0= 1+ \max\{1, P^t\}$ then the algorithm's expected revenue is at least 
\begin{align*}
(1-\delta)\cdot \max\{1, P^t\} & \geq (1-2\delta)\cdot \max\{1, P^t\} + \frac{\delta}{2} \cdot \left(1+\max\{1, P^t\}\right) \\
&\geq (1-2\delta)\cdot \max\{1, P^t\} + \frac{\delta}{2} \cdot v^t_{s^t_*}.
\end{align*}
Otherwise, let $k\in \{0,1,\ldots, \lfloor\log v\rfloor\}$ be such that,
$\max\{1, P^t\}+ 2^{k}\leq  \vt_{s^t_*} \leq \max\{1, P^t\}+ 2^{k+1}$.
The total expected revenue of the algorithm is:
\begin{align*}
&(1-\delta) \cdot \max\{1, P^t\}+ \frac{\delta}{1+\lfloor\log v\rfloor} \cdot \left(\max\{1, P^t\}+\sum_{i=0}^{k}2^i\right)\\
& \geq (1-\delta) \cdot \max\{1, P^t\}+ \frac{\delta}{1+\log v} \cdot \left(\max\{1, P^t\}+ 2^{k}\right) \\
&\geq (1-\delta) \cdot \max\{1, P^t\}+ \frac{\delta}{2+2\log v} \cdot v^t_{s^t_*}.
\end{align*}
The last inequality follows since $v^t_{s^t_*}\leq \max\{1, P^t\}+ 2^{k+1} \leq 2\cdot (\max\{1, P^t\}+ 2^{k}).$
\end{proof}
Plugging $\delta=\nicefrac{\eps}{8}$ into Lemma \ref{lem:revenue}, the algorithm's  expected revenue is at least, 
\begin{align*}
E[V_{alg}] & \geq \sum_{t\in \tbuy'}\left(\left(1-2\delta \right)\cdot \max\{1, P^t\} + \frac{\delta}{2(1+\log v)} \cdot v^t_{s^t_*}\right) - \sum_{t\in \tsell'}\frac{P^t}{1+\eps}\\
& = \sum_{t\in \tbuy'}\left(\left(1-\frac{\eps}{4} \right)\cdot \max\{1, P^t\} + \frac{\eps \cdot v^t_{s^t_*}}{16(1+\log v)} \right) - \sum_{t\in \tsell'}\frac{P^t}{1+\eps}.
 \end{align*} 

Combining this with the upper bound on the optimal profit we get,
\begin{align*}
OPT &= O\left(\frac{\eta}{\eps}\right) \cdot \left[\sum_{t\in \tbuy'}\left(\left(1-\frac{\eps}{4}\right)\cdot P^t +\frac{\eps \cdot \vt^t_{s_{*}^t}}{\eta}+ \frac{1}{\mu}\right) - \sum_{t\in \tsell'}\frac{P^t}{1+\eps}\right]\\
& = O\left(\frac{\eta}{\eps}\right) \cdot \left[\sum_{t\in \tbuy'}\left(\left(1-\frac{\eps}{4}\right)\cdot \max\{1, P^t\} +\frac{\eps \cdot \vt^t_{s_{*}^t}}{32(1+\log v)}+ \frac{\eps}{32(1+\log v)}\right) - \sum_{t\in \tsell'}\frac{P^t}{1+\eps}\right]\\
& = O\left(\frac{\eta}{\eps}\right) \left[\sum_{t\in \tbuy'}\left(\left(1-\frac{\eps}{4}\right) \max\{1, P^t\} +\frac{\eps \cdot \vt^t_{s_{*}^t}}{16 (1+\log v)}\right) - \sum_{t\in \tsell'}\frac{P^t}{1+\eps}\right] = O\left(\frac{\eta}{\eps}\right) E[V_{alg}],
\end{align*}
where the final inequality uses the fact that $\vt^t_{s_{*}^t}\geq 1$ for $t\in \tbuy$.
Finally, note that $\eta= O(\log(dv / \eps))$. Hence, the algorithm is $O\left( \eta / \eps\right)= O\left(\log(dv / \eps)/\eps\right)$-competitive concluding the proof of the second part of Theorem \ref{thm:main-formal}.

\section{Lower Bounds}\label{sec:lower}
In this section, we prove our lower bound theorem.

\begin{theorem}\label{thm:lower-bound-formal}
For the online trading problem (even for the known valuation setting and single minded customers/suppliers), when comparing to an optimal offline fractional solution for which supplier values  are $(1+\eps)$ larger:
\begin{itemize}
\item 
The competitive ratio of any deterministic or randomized algorithm is $\Omega( \frac{1}{\eps}\cdot \log (dv))$. This holds even if all the items are of a single type.
In particular, without the $1+\eps$ supplier value augmentation, the competitive ratio of any algorithm is unbounded.
   \item There exists a constant $c$, such that if the inventory from an item type is less than $\frac{c}{\eps}\cdot \log(dv)$ times the number of items of type $i$ in some of the bundles, then the competitive ratio of any {\bf deterministic} algorithm is unbounded. This holds even with respect to an optimal solution that can hold a single item from each type.
\end{itemize}
\end{theorem}

In Section \ref{sec:lower1} we prove the first statement of the theorem. In Section \ref{sec:lower2} we prove the second statement of the theorem.

\subsection{$\Omega(\frac{1}{\eps}\cdot \log{dv})$ Lower Bound}\label{sec:lower1}

In this section we prove that the competitive ratio of any deterministic or randomized algorithm
is $\Omega( \frac{1}{\eps}\cdot \log (dv))$, even in the special case where all items are of a single type.
In particular, without the $1+\eps$ value augmentation of the suppliers, the competitive ratio is unbounded. We prove two lower bounds separately: (a) $\Omega( \frac{1}{\eps}\cdot \log v)$ even when all customers and suppliers wish to buy or sell a single item ($d=1$), but the value for the item is in $[1,v]$ for some arbitrary value $v>1$; (b) $\Omega( \frac{1}{\eps}\cdot \log d)$ even when the value of the bundles requested in the range $[1,2]$ (i.e. $v=2$), but the bundles may contain up to $d$ items for an arbitrary $d\in \Z_+$. The two lower bounds are very similar, but we show them separately for clarity. 

Our bounds hold even when the algorithm is allowed to sell or buy items \emph{fractionally}. Observe that any randomized algorithm $\mathcal{R}$ for the trading problem (against an oblivious adversary!) induces a feasible \emph{fractional} solution $r$, where $r^t_i$ is the expected inventory item type $i$ that algorithm $\mathcal{R}$ holds at time $t$. Therefore, the lower bound we prove holds even for randomized algorithms. 

\textbf{Intuition:} \quad In the hard input sequence we construct, there are unbounded number of phases. In each phase, the adversary presents to the algorithm a stream of suppliers each selling at exponentially decreasing cost a full inventory's worth of the same item. As long as the algorithm purchases ``enough'' of the items, the phase continues. Otherwise, the adversary presents a set of customers offering to buy items at a price $(1+\eps)^2$ times the last (cheapest) supplier's price, and the phase ends. 
If the last suppliers arrive with a cheapest price per unit of item of $x$, then the optimal solution purchases a full inventory for a price per unit of $(1+\eps)x$ (paying $(1+\eps)$ times the price paid by the algorithm), and immediately sells at a price per unit of $(1+\eps)^2\cdot x$, making profit of $\eps\cdot (1+\eps) x$ per unit.  
On the other hand, we argue that we can define ``enough'' such that the algorithm needs to spend too much money over the course of the sequence to make more than a fraction of this optimal profit. We repeat this construction an arbitrary number of times in phases to amortize away any initialization constants.

In the $\Omega(\frac{1}{\eps} \log v)$ lower bound the decreasing prices per unit are achieved via suppliers with  decreasing values. In the $\Omega(\frac{1}{\eps} \log d)$ lower bound, the decreasing prices are instead achieved via suppliers with (roughly) fixed values but increasing sizes of bundles. There is an additional technical complexity in the second bound that stems from the fact that bundle sizes must be integral, and to achieve this we vary valuations slightly in the range $[1,2]$.

\subsubsection{Proof of the $\Omega(\frac{1}{\eps} \log v)$ lower bound}

\begin{lemma}
    The competitive ratio of any deterministic or randomized algorithm is $\Omega(\frac{1}{\eps}\cdot \log v)$. This holds even if $d = 1$, all the items are of a single type, and the customers/supplier are single minded.
\end{lemma}

\begin{proof}

We assume that $v$ is such that $v\geq (1+\eps)^8$, and let $c=\lfloor \nicefrac{1}{2}\cdot \log_{1+\eps}v\rfloor - 1$ be an integer (thus larger than 3).
The input sequence is divided into phases, each of which consists of an adaptive sequence of steps in which suppliers arrive, followed by one single step in which customers arrive.
We denote by $y_0 \in [0,w]$ the inventory of the algorithm at the beginning of the phase is (before step $0$), and let $y_t \in [y_0, w]$ be the (potentially fractional) inventory of the algorithm before the 
$t$-th step. Let $i_t = \lfloor\frac{y_t}{w} \cdot c \rfloor$ (so that $\frac{y_t}{w} \cdot c \in [i_t, i_t+1)$ and $i_t \in [0,c]$). 

At the $t$-th step, $w$ suppliers arrive, each offering to sell a single item with value $\frac{v}{(1+\eps)^{2(i_t+1)}}\in [\frac{v}{(1+\eps)^2},1]$. The algorithm may purchase some fraction of items from the suppliers, thus increasing its inventory to $y_{t+1}\geq y_t$. If $\frac{y_{t+1}}{w} \cdot c  \leq i_t+1$, then $w$ customers arrive, each wanting to purchase a single item with value $\frac{v}{(1+\eps)^{2i_t}} \in [v, (1+\eps)^2]$, and the phase ends. Otherwise we continue to step $t+1$. Let $F$ be the final time step in which suppliers arrive. Thus the inventory of the algorithm before the arrival of the last $w$ suppliers is $y_F$. These suppliers have value $\frac{v}{(1+\eps)^{2(i_F + 1)}}$, and thus the customers have value $\frac{v}{(1+\eps)^{2i_F}}$.

Note that every phase must eventually end because at each step in which no customer arrives (and the phase continues) we have $i_{t+1} \geq i_t + 1$, and for all $t$, we have $i_t \leq c$. Note also that due to our choice of $c=\lfloor \nicefrac{1}{2}\cdot \log_{1+\eps}v\rfloor - 1$, the values of all the suppliers and customer are indeed in the range $[1,v]$. 

Let $\Delta v_{alg}$ and $\Delta v_{adv}$ be the profit of the algorithm and the adversary in a single phase. To complete the proof of the lower bound, we compare these two quantities.

\textbf{Bounding} $\Delta v_{adv}$: \quad The adversary can buy $w$ items from the last $w$ suppliers and then immediately sell the entire inventory to the $w$ customers.\footnote{Technically, in the first phase, the adversary's inventory is full and it does not need to pay to fill its inventory, which only helps our analysis.  In every subsequent phase, the inventory of the adversary is initially empty.}  
To purchase the $w$ items, it pays $w \cdot (1+\eps) \cdot \frac{v}{(1+\eps)^{2(i_F+1)}}$ (this is $(1+\eps)$ times the price offered to the algorithm by the suppliers in the last step). Hence, its total trading profit in the phase is 
\[\Delta v_{adv} = w\cdot \left(\frac{v}{(1+\eps)^{2i_F}} - (1+\eps)\cdot \frac{v}{(1+\eps)^{2(i_F+1)}}\right) = w\cdot \frac{\eps \cdot v}{(1+\eps)^{2i_F+1}}.\]

\textbf{Bounding} $\Delta v_{alg}$: \quad To analyze the profit of \alg, we imagine that it represents the inventory it holds as a subset of the interval $[0,w]$. We further imagine that when buying, it fills the interval $[0,w]$ from left to right, and when it sells it clears inventory from right to left (i.e. LIFO, see \cref{fig:interval}). This is purely for accounting purposes and will not change the total profit of the algorithm. Next, we partition the inventory in the interval $[0,w]$ into $c$ (sub-)intervals indexed by $j=0,1, \ldots, c-1$, where the $j$th interval is the inventory between $[\frac{j}{c} \cdot w , \frac{j+1}{c} \cdot w]$. The algorithm maintains the following invariant.\footnote{Note that the algorithm did not pay for its starting inventory: to account for this, we can imagine the algorithm starts the game with a profit of $\sum_{j = 0}^{c-1} \frac{w}{c} \frac{v}{(1+\eps)^{2(j+1)}}$ which it then pays to ensure the invariant holds. This adds an absolute constant to the profit of the algorithm, which can be amortized to $0$ after sufficiently many phases.}

\begin{invariant}
The price paid per unit for the $j$-th interval is at least $\frac{v}{(1+\eps)^{2(j+1)}}$.
\end{invariant}

\begin{proof}
    Since the algorithm fills its inventory $[0,w]$ according to the LIFO policy, the occupied inventory at time $t$ is always the interval $[0,y_t]$.    Suppose that $[y_t, y_{t+1}] \cap [\frac{j}{c} \cdot w , \frac{j+1}{c} \cdot w] \neq \emptyset$, in other words the algorithm partially fills the interval $j$ in time step $t$. Then $\frac{y_t}{w} \cdot c < j+1$, so by construction the suppliers in time step $t$ have value at least $\frac{v}{(1+\eps)^{2(j + 1)}}$, and hence the fraction if the $j$-th interval covered by $[y_t,y_{t+1}]$ is bought at no less than this price.
\end{proof}

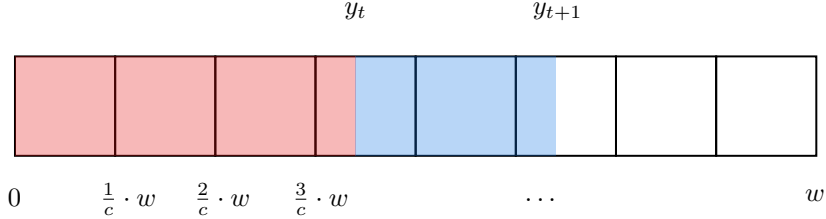
\begin{figure}[H]
	\centering
	 \tikzset{every picture/.style={line width=0.75pt}} 

\begin{tikzpicture}[x=0.75pt,y=0.75pt,yscale=-1,xscale=1]

\draw   (80,110) -- (130,110) -- (130,160) -- (80,160) -- cycle ;
\draw   (130,110) -- (180,110) -- (180,160) -- (130,160) -- cycle ;
\draw   (180,110) -- (230,110) -- (230,160) -- (180,160) -- cycle ;
\draw   (230,110) -- (280,110) -- (280,160) -- (230,160) -- cycle ;
\draw   (280,110) -- (330,110) -- (330,160) -- (280,160) -- cycle ;
\draw   (330,110) -- (380,110) -- (380,160) -- (330,160) -- cycle ;
\draw   (380,110) -- (430,110) -- (430,160) -- (380,160) -- cycle ;
\draw   (430,110) -- (480,110) -- (480,160) -- (430,160) -- cycle ;
\draw  [draw opacity=0][fill={rgb, 255:red, 230; green, 18; blue, 18 }  ,fill opacity=0.3 ] (80,110) -- (250,110) -- (250,160) -- (80,160) -- cycle ;
\draw  [draw opacity=0][fill={rgb, 255:red, 18; green, 119; blue, 230 }  ,fill opacity=0.3 ] (250,110) -- (350,110) -- (350,160) -- (250,160) -- cycle ;

\draw (75,175.4) node [anchor=north west][inner sep=0.75pt]    {$0$};
\draw (121,172.4) node [anchor=north west][inner sep=0.75pt]    {$\frac{1}{c}\cdot w$};
\draw (168,172.4) node [anchor=north west][inner sep=0.75pt]    {$\frac{2}{c} \cdot w$};
\draw (217,172.4) node [anchor=north west][inner sep=0.75pt]    {$\frac{3}{c} \cdot w$};
\draw (473,175.4) node [anchor=north west][inner sep=0.75pt]    {$w$};
\draw (333,179.8) node [anchor=north west][inner sep=0.75pt]    {$\dotsc $};
\draw (243,81.4) node [anchor=north west][inner sep=0.75pt]    {$y_{t}$};
\draw (337,81.4) node [anchor=north west][inner sep=0.75pt]    {$y_{t+1}$};

\end{tikzpicture}
	\caption{Illustration of the accounting scheme. We imagine the algorithm fills its inventory from left to right. By construction, the price at which the algorithm can fill any fraction of the $j$-th interval is at least $\frac{v}{(1+\eps)^{2(j+1)}}$. \label{fig:interval}}
\end{figure}

To conclude the proof, recall $F$ is the final time step in which suppliers arrive, and the inventory of the algorithm before the arrival of the last $w$ suppliers is $y_F$. As the phase ends this means that $\frac{y_{F+1}}{w} \cdot c  \leq i_F+1$ meaning that $y_{F+1}\leq \frac{i_{F}+1}{c}\cdot w$ (the algorithm did not fill more than a single sub-interval). The profit the algorithm can make from selling items in the interval $[\frac{i_F}{c} \cdot w, y_{F+1}]$ is at most
\begin{align*}
 &\left(y_{F+1}-  \frac{i_F}{c}\cdot w\right)\cdot \left(\frac{v}{(1+\eps)^{2i_F}}- \frac{v}{(1+\eps)^{2(i_F+1)}}\right) \leq \frac{w}{c}\cdot \left(\frac{v}{(1+\eps)^{2i_F}}- \frac{v}{(1+\eps)^{2(i_F+1)}}\right) \\
 &= \frac{w}{c}\cdot\frac{v}{(1+\eps)^{2i_F+1}} \left(1+\eps- \frac{1}{1+\eps}\right)\leq \frac{2w}{c} \cdot \frac{\eps \cdot v}{(1+\eps)^{2i_F+1}}= \frac{2}{c}\cdot \Delta v_{adv}
\end{align*}
If the algorithm chooses to further sell inventory the range $[0, \frac{i_F}{c} \cdot w]$, which was purchased at a price per item of at least $\frac{v}{(1+\eps)^{2i_F}}$, the algorithm makes no profit (and will even lose money for selling from the range $[0, \frac{i_F-1}{c} \cdot w]$).

We conclude that the competitive ratio of the algorithm is at least $c/2=\Omega(\frac{1}{\eps}\log v)$.
\end{proof}

\subsubsection{Proof of the $\Omega(\frac{1}{\eps} \log d)$ lower bound}

\begin{lemma}
    The competitive ratio of any deterministic or randomized algorithm is $\Omega(\frac{1}{\eps}\cdot \log d)$. This holds even if $v = 2$, all the items are of a single type, and the customers/supplier are single minded.
\end{lemma}

\begin{proof}
This time we construct an instance for every $d \geq 2$. We assume that $d$ is a power of $2$ that divides $w$, and $\eps$ is such that $d \geq (1+\eps)^8$. Let
$c=-1+\lfloor \frac{1}{2}\cdot \log_{1+\eps}d\rfloor$ be an integer, which by this assumption is at least $3$.

Again, the input is divided into phases, each of which consists of an adaptive sequence of suppliers, followed by one set of customers. 
We denote by $y_0 \in [0,w]$ the inventory of the algorithm at the beginning of the phase (before step $0$), and let $y_t \in [y_0, w]$ be the (potentially fractional) inventory of the algorithm before the $t$-th step. Let $i_t = \lfloor\frac{y_t}{w} \cdot c \rfloor$ (so that $\frac{y_t}{w} \cdot c \in [i_t, i_t+1)$ and $i_t \in [0,c]$).  Additionally define $x_t = \frac{1}{(1+\eps)^{2i_t}}$. By our choice of $c=-1+\lfloor \frac{1}{2}\cdot \log_{1+\eps}d\rfloor$, we have $1\leq \frac{1}{x_t}\leq d$, for $t=0, \ldots, c+1$.
Let $d_t$ be the smallest power of $2$ greater than $1/x_t$, and define $v_t = x_t \cdot d_t \in [1,2]$. We have $d_t\in [1,d]$.

At the $t$-th step, $w/d_{t+1}$ suppliers arrive, each offering to sell $d_{t+1}$ items at value $v_{t+1}$. Note $d$ divides $w$ and $d_t$ is also a power of 2 less than $d$, so $\frac{w}{d_t}$ is integral. The algorithm may purchase some fraction of items from the suppliers, thus increasing its inventory to $y_{t+1} \geq y_t$. If $\frac{y_{t+1}}{w} \cdot c \leq i_t + 1$, then $\frac{w}{d_{t}}$ customers arrive, each offering to buy a bundle of $d_{t}$ items for a price of $v_{t}$, and the phase ends. Otherwise we continue to step $t+1$. Note that every phase must eventually end because at each step in which no customer arrives (and the phase continues) we have $i_{t+1} \geq i_t + 1$, and for all $t$ we have $i_t \leq c$. Once again let $F$ be the final time step in which suppliers arrive. 

\textbf{Bounding} $\Delta v_{adv}$: \quad The adversary can buy $d_{F+1}$ items from each of the last $\frac{w}{d_{F+1}}$ suppliers, thus buying $w$ items in total and filling its inventory.\footnote{Once again, in the very first phase the adversary starts with a full inventory and does not have to make any purchases.} Then it can sell $d_{F}$ items to each of the $\frac{w}{d_{F}}$ customers, thus selling $w$ items in total and clearing its inventory. Its total profit (it pays $1+\eps$ times more than the algorithm to the suppliers) is:
\[\Delta v_{adv} = \frac{w}{d_{F}} \cdot v_{F} - (1+\eps) \frac{w}{d_{F+1}} \cdot v_{F+1} = w\left(\frac{1}{(1+\eps)^{2i_F}} -  \frac{1+\eps}{(1+\eps)^{2(i_F+1)}}\right) = w\cdot \frac{\eps}{(1+\eps)^{2i_F+1}},\]
where we used that $v_{t} / d_t = x_t = (1+\eps)^{-2 i_t}$ by definition.

\textbf{Bounding} $\Delta v_{alg}$: Once again, to analyze the profit of \alg we imagine that it fills the interval $[0,w]$ from left to right when buying, and clears it from right to left when selling. We again partition $[0,w]$ into $c$ (sub-)intervals indexed by $j=0,1, \ldots, c-1$, where the $j$th interval is the inventory between $[\frac{j}{c} \cdot w , \frac{j+1}{c} \cdot w]$. The invariant we maintain this time is very similar to the one before.\footnote{The algorithm does not pay for its initial inventory, but this adds a constant to the algorithm's profit which amortizes to $0$ after enough phases.}

\begin{invariant}
The price paid per unit for the $j$-th interval is at least $\frac{1}{(1+\eps)^{2(j+1)}}$.
\end{invariant}
\begin{proof}
    The LIFO policy ensures that the occupied portion of the inventory $[0,w]$ at time step $t$ is the interval $[0,y_t]$. Suppose that  $[y_t, y_{t+1}] \cap [\frac{j}{c} \cdot w , \frac{j+1}{c} \cdot w] \neq \emptyset$ (the algorithm fills some part of the $j$th interval at time $t$). Then $\frac{y_t}{w} \cdot c < j+1$, so the suppliers time step $t$ offer a bundle of $d_{t+1}$ items with value $v_{t+1}$ such that the price per unit is $\frac{v_{t+1}}{d_{t+1}}=x_{t+1} = \frac{1}{(1+\eps)^{2(i_t+1)}} \geq \frac{1}{(1+\eps)^{2(j+1)}}$. This is precisely the price per unit paid for the portion of the $j$-th interval covered in this step.
\end{proof}

Hence, when the $\frac{w}{d_{F}}$ customers arrive at the end of the phase with wishing to buy a bundle of size $d_{F}$ at price $v_{F}$, their value per item is $\frac{v_{F}}{d_{F}}= x_{F}=\frac{1}{(1+\eps)^{2i_F}}$. By construction, the customers arrive at time $t=F+1$ when $\frac{y_{F+1}}{w}\cdot c \leq i_F+1$ ($y_{F+1} \leq w \cdot \frac{i_F+1}{c}$), so the profit of the algorithm from selling the cheapest items in the interval $[\frac{i_F}{c}\cdot w, y_{F+1}]$ is at most, 
\begin{align*}
 &\left(y_{F+1}-  \frac{w}{c}\cdot i_F\right)\cdot \left(\frac{1}{(1+\eps)^{2i_F}}- \frac{1}{(1+\eps)^{2\cdot (i_F+1)}}\right) \leq \frac{w}{c}\cdot \left(\frac{1}{(1+\eps)^{2i_F}}- \frac{1}{(1+\eps)^{2\cdot (i_F+1)}}\right) \\
 &= \frac{w}{c}\cdot\frac{1}{(1+\eps)^{2i_F+1}} \left(1+\eps- \frac{1}{1+\eps}\right)\leq \frac{2w}{c} \cdot \frac{\eps}{(1+\eps)^{2i_F+1}}= \frac{2}{c}\cdot \Delta v_{adv}
\end{align*}
If the algorithm chooses to further sell inventory the range $[0, \frac{i_F}{c} \cdot w]$, which was purchased at a price per item of at least $\frac{v}{(1+\eps)^{2i_F}}$, the algorithm makes no profit (and will even lose money for selling from the range $[0, \frac{i_F-1}{c} \cdot w]$).

We conclude that the competitive ratio of the algorithm is at least $c/2=\Omega(\frac{1}{\eps}\log d)$.
\end{proof}
\subsection{Unbounded Competitiveness when the Inventory is too Small}\label{sec:lower2}

In this section we prove that there exists a constant $c$, such that if the inventory of an item type is smaller than $\frac{c}{\eps}\cdot \log(dv)$ times the number of items of type $i$ in some of the bundles, then the competitive ratio of any {\bf deterministic} algorithm is unbounded.
In Section \ref{sec:lower:unbounded1} we prove that the competitive ratio of any deterministic algorithm is unbounded for a single item type, arbitrary $v>1$, and when all bundle requests are for a single item ($d=1$) whenever the inventory is strictly less than $\frac{c}{\eps}\cdot \log v$ for some constant $c$. 
In Section \ref{sec:lower:unbounded2} we look at an instance with (at least) $d$ items types for an arbitrarily large $d$ which is a power of 2, $v=8$ is a constant, and all bundles consist of a at most a single item from each item type.
We prove that whenever the inventory is strictly less than $\frac{c}{\eps}\cdot \log d$ for some constant $c$, the competitive ratio of any deterministic algorithm is unbounded.
The combination of these two bounds prove our claim. 

\subsubsection{Unbounded competitive ratio when inventory is too small compared to $v$ and $d=1$}\label{sec:lower:unbounded1}

\begin{lemma}
    \label{lem:unboundedlowerbound}
    If the inventory from an item type is less than $\frac{1}{4\eps}\cdot \log(v)$ times the number of items of type $i$ in some of the bundles, then the competitive ratio of any {\bf deterministic} algorithm is unbounded. This holds also when all customers/suppliers are single minded, and even with respect to an optimal solution that can hold a single item from each type.
\end{lemma}

\begin{proof}
Assume that $v\geq (1+\eps)^4$, and that there is a single item type. If the inventory cap for this item is $w\leq \frac{1}{4} \log_{1+\eps} v$, then this means that $w\leq \frac{1}{4} \log_{1+\eps} v +(-1+\frac{1}{4} \log_{1+\eps} v )\leq -1+\frac{1}{2} \log_{1+\eps} v$.

The input is divided into phases as in previous sections, where each phase is the following adaptive sequence. Denote by $Y_0 \in [0,w]$ the inventory of the algorithm at the beginning of the phase (before step $0$) which we assume is \emph{integral} this time (as the algorithm is deterministic). Let $Y_t$ denote the integral inventory of the algorithm before step $t$.

At the $t$-th step, a single supplier arrives offering to sell a single unit (and no more!) at value $\frac{v}{(1+\eps)^{2(Y_t+1)}}$. If the (deterministic) algorithm decides to buy the item, then the algorithm increases its inventory to $Y_{t+1} = Y_t+1$ and the phase continues to step $t+1$. Otherwise, a single customer arrives wishing to purchase a single unit (and no more) at value $\frac{v}{(1+\eps)^{2Y_t}}\in [(1+\eps)^2,v]$, and the phase ends. Let $F$ be the final step in which suppliers arrive, such that the last supplier has value $\frac{v}{(1+\eps)^{2(Y_F+1)}}$, and the customer has value $\frac{v}{(1+\eps)^{2Y_F}}$.

Note that every phase must end because in every step $t$ that it continues, $Y_t$ increases by $1$, and we have $Y_t \leq w$. Also observe that customer values indeed lie in the range $[1,v]$ because these are at most $\frac{v}{(1+\eps)^{2}}$ (when $Y_t=0$) and at least $\frac{v}{(1+\eps)^{2(w+1)}}\geq 1$ (when $Y_t=w)$, since $w\leq -1+ \frac{1}{2} \log_{1+\eps} v$. 

\textbf{Bounding} $\Delta v_{adv}$: \quad The adversary can buy the item from the last supplier at time step $F$ at a price of $(1+\eps) \frac{v}{(1+\eps)^{2(Y_F+1)}} = \frac{v}{(1+\eps)^{2Y_F+1}}$,\footnote{Once again, the adversary does not need to buy in the very first phase, which only improves our analysis.} and immediately sell the item to the customer at a price of $\frac{v}{(1+\eps)^{2Y_F}}$. Hence it's total trading profit per phase is
\[\Delta v_{adv} = \frac{v}{(1+\eps)^{2Y_F}} - \frac{v}{(1+\eps)^{2Y_F+1}} = \frac{\eps \cdot v}{(1+\eps)^{2Y_F+1}} > 0.\]

\textbf{Bounding} $\Delta v_{alg}$: We assume the algorithm fills and empties its inventory according to the LIFO policy, as in the previous sections. Since we assume its inventory is integral, we can prove a relatively simple invariant this time.\footnote{The algorithm does not pay for its initial inventory, but this adds a constant to the algorithm's profit which amortizes to $0$ after enough phases.}

\begin{invariant}
The price paid the $j$-th unit of item is at least $\frac{v}{(1+\eps)^{2j}}$.
\end{invariant}

\begin{proof}
    The LIFO policy ensures that the occupied portion of the inventory $[0,w]$ at time step $t$ is the interval $[0,Y_t]$. If the algorithm purchases an item in time step $t$, then this item is the $(Y_t+1)$-th unit and costs $\frac{v}{(1+\eps)^{2(Y_t + 1)}}$.
\end{proof}

Since the algorithm holds $Y_F$ items when the phase ends, the cheapest item it holds has cost $\frac{v}{(1+\eps)^{2Y_F}}$, which is also the value of the customer in this phase. Hence the algorithm can at best make profit $\Delta v_{alg} = 0$.

Since the adversary's profit is strictly positive and the algorithm's is $0$, we conclude that the competitive ratio of the algorithm is unbounded.
\end{proof}

\subsubsection{Unbounded competitive ratio when inventory is too small compared to $d$ and $v=8$.}\label{sec:lower:unbounded2}

\begin{lemma}
    \label{lem:unboundedlowerbound_2}
    If the inventory from an item type is less than $\frac{1}{8\eps}\cdot \log d$ times the number of items of type $i$ in some of the bundles, then the competitive ratio of any {\bf deterministic} algorithm is unbounded. This holds also when all customers/suppliers are single minded, and even with respect to an optimal solution that can hold a single item from each type.
\end{lemma}

\begin{proof}

We assume $d\geq (1+\eps)^8$ and that $d$ is a power of $2$. Our instance has $d$ item types each with the same inventory cap of $w$. We will have $v=8$. Let $c=-1+\lfloor \frac{1}{2}\cdot \log_{1+\eps}d\rfloor\geq 3$ be an integer.
If the inventory cap $w$ is such that $w\leq \frac{1}{8\eps}\cdot \log d \leq \frac{1}{8} \log_{1+\eps} d$, then
$w \leq \frac{1}{8} \log_{1+\eps} d + (-1 + \frac{1}{8} \log_{1+\eps} d) \leq -1 + \frac{1}{4} \log_{1+\eps} d \leq c$.

Yet again, the input is divided into phases, but for a change, our phases are short. Each consists of either a single supplier, or a single supplier followed by a single customer. We start by defining the (integral!) inventory of the algorithm at time $t$ to be the vector $Y^t \in [0,w]^d$. Now for every $\ell \in \{0\} \cup [w]$ the quantity $x_\ell= \frac{1}{(1+\eps)^{2(\ell+1)}}\in [\frac{1}{(1+\eps)^2},\frac{1}{d}]$. Let $d_\ell$ be the smallest power of $2$ greater than $1/x_\ell$, and define $v_\ell = x_\ell \cdot d_\ell \in [1,2]$.
Note, that as $x_{\ell}$ are decreasing as $\ell$ increases, then $d_{\ell}$ increases as $\ell$ increases.
Also define $x'_\ell = (1+\eps)^2 x_\ell$, and set $v'_t = x'_\ell d_\ell$. Since $x'_\ell\leq 4x_t$, we have $v'_\ell \in [1,8]$. Finally, define $\mathcal{S}_\ell^t$ to be the set of item types that the algorithm holds exactly $\ell$ of in inventory at time $t$. Let $k(t) = \min \{\ell \mid \mathcal{S}_\ell^t \neq \emptyset\}$, that is the smallest number of units in the algorithm's inventory over all item types.

At the $t$-th step, a single supplier arrives wishing to sell a bundle $S^t \subseteq \mathcal{S}^t_{k(t)}$ at value $v_{k(t)}$. We require that $|S^t| = d_{k(t)}$ and that $S^t$ contain at most one of every item type, but otherwise $S^t$ is arbitrary. 
If the (deterministic) algorithm decides to buys the bundle, the phase ends. Otherwise, a single customer arrives that would like to purchase the subset $S$ at value $v_{k(t)}'$, and then the phase ends.

The astute reader may worry: why can we guarantee that there are at least $d_{k(t)}$ items with exactly $k(t)$ copies in the algorithm's inventory? We resolve this issue with the following lemma.

\begin{lemma}\label{lem:sjt_structure}
    At any time $t$, and for any $i\in[0,w]$, $d_i$ divides $\sum_{j\leq i}|S_j^t|$.
\end{lemma}

In particular, $d_{k(t)}$ divides $|S_{k(t)}^t|=\sum_{j\leq k(t)}|S_j^t|$. We will prove \cref{lem:sjt_structure} soon, but let us first see how to finish the argument.

Because the algorithm's inventory is finite, the number of phases of length $2$ goes to infinity with the number of phases.

\textbf{Bounding} $\Delta v_{adv}$: \quad The adversary does nothing in phases of length $1$. In phases of length $2$, it buys bundle $S$ at value $(1+\eps) v_{k(t)}$ and sells it at value $v'_{k(t)} = (1+\eps)^2 v_{k(t)}$. Hence its profit in phases of length $2$ is
$\eps (1+\eps) v_{k(t)} > 0$, and hence its profit goes to infinity with the number of phases.

\textbf{Bounding} $\Delta v_{alg}$: \quad To analyze the algorithm, as usual we imagine that it maintains its inventory in LIFO fashion: purchases happen from left to right, sales happen from right to left. We also imagine dividing the cost of every bundle purchased by the algorithm \emph{uniformly} among the items in the bundle (thus associating a cost with each unit of item).\footnote{Once again, the algorithm does not pay for its initial inventory, but this adds a constant term to the algorithm's profit which amortizes to $0$.} We then argue that every bundle sold makes less revenue for the algorithm than the sum of costs of item units it contains. We maintain the following invariant.

\begin{invariant}
 The price paid the $j$-th unit of any item type is at least $\frac{1}{(1+\eps)^{2j}}$.
\end{invariant}

\begin{proof}
    The LIFO policy ensures that for each item type $i$, the occupied portion of the inventory $[0,w]$ at time step $t$ is the interval $[0,Y^t_i]$. If the algorithm purchases a unit of item $i$ in time step $t$, it pays $v_{k(t)} = x_{k(t)} d_{k(t)}$ for $d_{k(t)}$ units in the entire bundle, so $x_{k(t)} = \frac{1}{(1+\eps)^{2(Y^t_i + 1)}}$ per unit of the bundle. The item in question is the the $(Y^t_i+1)$-th unit of item type $i$, so the invariant holds.
\end{proof}

Finally we can conclude the proof. Suppose the algorithm sells bundle $S^t \subseteq \mathcal{S}^t_{k(t)}$ at time $t$ at value $v'_{k(t)} = x'_{k(t)} d_{k(t)} = \frac{d_{k(t)}}{(1+\eps)^{2k(t)}}$. By construction, this bundle consists of the $k(t)$-th unit of $d_{k(t)}$ distinct item types. By the invariant above, each unit was bought at a cost of $\frac{v}{(1+\eps)^{2k(t)}}$, and so the total cost to purchase the items in this bundle is also $\frac{d_{k(t)}}{(1+\eps)^{2k(t)}}$. Hence the algorithm makes a profit of at most $0$.

We conclude with the missing proof, which ensures that our construction is well-defined.

\begin{proof}[Proof of \cref{lem:sjt_structure}]
    The proof is by induction on time steps $t$. The property holds at time $t=0$ when the inventory of the algorithm is full as we assumed that $d$ is a power of 2, and every $d_\ell$ is a power of 2.
    
    Assume inductively that the property holds before time $t$. There are three cases to consider. \begin{enumerate}
        \item If the algorithm decides not to purchase the bundle $S^t$ from the supplier, and also does not sell the bundle $S^t$ to the arriving customer, then the inventory is unchanged, and the invariant holds inductively. 
        \item   If the algorithm decides to purchase the bundle $S^t$ (and no customer arrives), then $|\mathcal{S}^{t+1}_{k(t)}| = |\mathcal{S}^t_{k(t)}| - d_{k(t)}$ (this may be of size 0 now), and $|\mathcal{S}^{t+1}_{k(t)+1}| = |\mathcal{S}^t_{k(t)+1}| + d_{k(t)}$. Thus, $\sum_{j\leq i} |\mathcal{S}^{t+1}_{j}|$ remains unchanged for all $i\neq k(t)$, and by our induction hypothesis $d_i$ divides $\sum_{j\leq i} |\mathcal{S}^{t+1}_{j}|$. For $i=k(t)$, $\sum_{j\leq k(t)} |\mathcal{S}^{t+1}_{j}| = \sum_{j\leq k(t)} |\mathcal{S}^{t}_{j}|- d_{k(t)}$, and therefore $d_{k(t)}$ still divides it.     
        \item Finally, if the algorithm decides not to purchase bundle $S^t$ from the supplier, but does choose to sell it to the arriving customer, then by definition $|\mathcal{S}^{t}_{k(t)-1}|=0$. Therefore,  $|\mathcal{S}^{t+1}_{k(t)-1}|=d_{k(t)}$ and $|\mathcal{S}^{t+1}_{k(t)}|=|\mathcal{S}^{t}_{k(t)}|- d_{k(t)}$. Therefore, $\sum_{j\leq i} |\mathcal{S}^{t+1}_{j}|$ remains unchanged for all $i\neq k(t)-1$, and by our induction hypothesis $d_i$ divides $\sum_{j\leq i} |\mathcal{S}^{t+1}_{j}|$.
However, as we observed, $d_i$ is a power of $2$ that is increasing as $i$ increases. Therefore $d_{k(t)-1}$ divides $d_{k(t)}$ and hence $d_{k(t)-1}$ divides $\sum_{j\leq k(t)-1}|\mathcal{S}^{t+1}_{j}| = d_{k(t)}$.
    \end{enumerate}

Recall that there is no fourth case because if the algorithm purchases bundle $S^t$, the phase ends.
\end{proof}
This concludes the proof of \cref{lem:unboundedlowerbound_2}.
\end{proof}

\bibliographystyle{plain}
\bibliography{refs}

\begin{thebibliography}{10}

\bibitem{AwerbuchAM03}
Baruch Awerbuch, Yossi Azar, and Adam Meyerson.
\newblock Reducing truth-telling online mechanisms to online optimization.
\newblock In Lawrence~L. Larmore and Michel~X. Goemans, editors, {\em
  Proceedings of the 35th Annual {ACM} Symposium on Theory of Computing, June
  9-11, 2003, San Diego, CA, {USA}}, pages 503--510. {ACM}, 2003.

\bibitem{AAP93}
Baruch Awerbuch, Yossi Azar, and Serge~A. Plotkin.
\newblock Throughput-competitive on-line routing.
\newblock In {\em 34th Annual Symposium on Foundations of Computer Science},
  pages 32--40. {IEEE} Computer Society, 1993.

\bibitem{BabaioffFN24}
Moshe Babaioff, Amitai Frey, and Noam Nisan.
\newblock Learning to maximize gains from trade in small markets.
\newblock In Dirk Bergemann, Robert Kleinberg, and Daniela Sab{\'{a}}n,
  editors, {\em Proceedings of the 25th {ACM} Conference on Economics and
  Computation, {EC} 2024, New Haven, CT, USA, July 8-11, 2024}, page 195.
  {ACM}, 2024.

\bibitem{BachocCCC24}
Fran{\c{c}}ois Bachoc, Nicol{\`{o}} Cesa{-}Bianchi, Tommaso Cesari, and Roberto
  Colomboni.
\newblock Fair online bilateral trade.
\newblock In {\em Advances in Neural Information Processing Systems 38: Annual
  Conference on Neural Information Processing Systems 2024, NeurIPS 2024,
  Vancouver, BC, Canada, December 10 - 15, 2024}, 2024.

\bibitem{BartalGN03}
Yair Bartal, Rica Gonen, and Noam Nisan.
\newblock Incentive compatible multi unit combinatorial auctions.
\newblock In Joseph~Y. Halpern and Moshe Tennenholtz, editors, {\em Proceedings
  of the 9th Conference on Theoretical Aspects of Rationality and Knowledge
  (TARK-2003), Bloomington, Indiana, USA, June 20-22, 2003}, pages 72--87.
  {ACM}, 2003.

\bibitem{Global-Budget-regret}
Martino Bernasconi, Matteo Castiglioni, Andrea Celli, and Federico Fusco.
\newblock No-regret learning in bilateral trade via global budget balance.
\newblock In {\em Proceedings of the 56th Annual ACM Symposium on Theory of
  Computing}, STOC 2024, page 247–258. Association for Computing Machinery,
  2024.

\bibitem{BLS23}
Sayan Bhattacharya, Niv Buchbinder, Roie Levin, and Thatchaphol Saranurak.
\newblock Chasing positive bodies.
\newblock In {\em 64th {IEEE} Annual Symposium on Foundations of Computer
  Science, {FOCS} 2023, Santa Cruz, CA, USA, November 6-9, 2023}, pages
  1694--1714. {IEEE}, 2023.

\bibitem{BlumrosenD21}
Liad Blumrosen and Shahar Dobzinski.
\newblock (almost) efficient mechanisms for bilateral trading.
\newblock {\em Games Econ. Behav.}, 130:369--383, 2021.

\bibitem{BE98}
Allan Borodin and Ran El{-}Yaniv.
\newblock {\em Online computation and competitive analysis}.
\newblock Cambridge University Press, 1998.

\bibitem{BuchbinderG15}
Niv Buchbinder and Rica Gonen.
\newblock Incentive compatible mulit-unit combinatorial auctions: {A} primal
  dual approach.
\newblock {\em Algorithmica}, 72(1):167--190, 2015.

\bibitem{BuchbinderN06}
Niv Buchbinder and Joseph Naor.
\newblock Improved bounds for online routing and packing via a primal-dual
  approach.
\newblock In {\em 47th Annual {IEEE} Symposium on Foundations of Computer
  Science {(FOCS} 2006), 21-24 October 2006, Berkeley, California, USA,
  Proceedings}, pages 293--304. {IEEE} Computer Society, 2006.

\bibitem{BN09}
Niv Buchbinder and Joseph Naor.
\newblock The design of competitive online algorithms via a primal-dual
  approach.
\newblock {\em Found. Trends Theor. Comput. Sci.}, 3(2-3):93--263, 2009.

\bibitem{Cesa-BianchiCCF24}
Nicol{\`{o}} Cesa{-}Bianchi, Tommaso Cesari, Roberto Colomboni, Federico Fusco,
  and Stefano Leonardi.
\newblock Bilateral trade: {A} regret minimization perspective.
\newblock {\em Math. Oper. Res.}, 49(1):171--203, 2024.

\bibitem{CDHOS23}
Jos{\'{e}} Correa, Andr{\'{e}}s Cristi, Paul Duetting, MohammadTaghi
  Hajiaghayi, Jan Olkowski, and Kevin Schewior.
\newblock Trading prophets.
\newblock In Kevin Leyton{-}Brown, Jason~D. Hartline, and Larry Samuelson,
  editors, {\em Proceedings of the 24th {ACM} Conference on Economics and
  Computation, {EC} 2023, London, United Kingdom, July 9-12, 2023}, pages
  490--510. {ACM}, 2023.

\bibitem{den2015dynamic}
Arnoud~V Den~Boer.
\newblock Dynamic pricing and learning: historical origins, current research,
  and new directions.
\newblock {\em Surveys in operations research and management science},
  20(1):1--18, 2015.

\bibitem{DobzinskiNS12}
Shahar Dobzinski, Noam Nisan, and Michael Schapira.
\newblock Truthful randomized mechanisms for combinatorial auctions.
\newblock {\em J. Comput. Syst. Sci.}, 78(1):15--25, 2012.

\bibitem{EFKT01}
Ran El{-}Yaniv, Amos Fiat, Richard~M. Karp, and G.~Turpin.
\newblock Optimal search and one-way trading online algorithms.
\newblock {\em Algorithmica}, 30(1):101--139, 2001.

\bibitem{FGL15}
Michal Feldman, Nick Gravin, and Brendan Lucier.
\newblock Combinatorial auctions via posted prices.
\newblock In Piotr Indyk, editor, {\em Proceedings of the Twenty-Sixth Annual
  {ACM-SIAM} Symposium on Discrete Algorithms, {SODA} 2015, San Diego, CA, USA,
  January 4-6, 2015}, pages 123--135. {SIAM}, 2015.

\bibitem{Fung19}
Stanley P.~Y. Fung.
\newblock Optimal online two-way trading with bounded number of transactions.
\newblock {\em Algorithmica}, 81(11-12):4238--4257, 2019.

\bibitem{Fung21}
Stanley P.~Y. Fung.
\newblock Online two-way trading: Randomization and advice.
\newblock {\em Theor. Comput. Sci.}, 856:41--50, 2021.

\bibitem{GT19}
Guillermo Gallego and Huseyin Topaloglu.
\newblock {\em Revenue Management and Pricing Analytics}.
\newblock Springer New York, NY, 2019.

\bibitem{GallegoR97}
Guillermo Gallego and Garrett~J. van Ryzin.
\newblock A multiproduct dynamic pricing problem and its applications to
  network yield management.
\newblock {\em Oper. Res.}, 45(1):24--41, 1997.

\bibitem{Jasin14}
Stefanus Jasin.
\newblock Reoptimization and self-adjusting price control for network revenue
  management.
\newblock {\em Oper. Res.}, 62(5):1168--1178, 2014.

\bibitem{LaviS11}
Ron Lavi and Chaitanya Swamy.
\newblock Truthful and near-optimal mechanism design via linear programming.
\newblock {\em J. {ACM}}, 58(6):25:1--25:24, 2011.

\bibitem{LM99}
Stefano Leonardi and Alberto Marchetti{-}Spaccamela.
\newblock On-line resource management with application to routing and
  scheduling.
\newblock {\em Algorithmica}, 24(1):29--49, 1999.

\bibitem{LH14}
Bin Li and Steven C.~H. Hoi.
\newblock Online portfolio selection: {A} survey.
\newblock {\em {ACM} Comput. Surv.}, 46(3):35:1--35:36, 2014.

\bibitem{abs-2403-05378}
Will Ma, Calum MacRury, and Jingwei Zhang.
\newblock Online contention resolution schemes for network revenue management
  and combinatorial auctions.
\newblock {\em CoRR}, abs/2403.05378, 2024.

\bibitem{MaRST20}
Yuhang Ma, Paat Rusmevichientong, Mika Sumida, and Huseyin Topaloglu.
\newblock An approximation algorithm for network revenue management under
  nonstationary arrivals.
\newblock {\em Oper. Res.}, 68(3):834--855, 2020.

\bibitem{MaglarasM06}
Constantinos Maglaras and Joern Meissner.
\newblock Dynamic pricing strategies for multiproduct revenue management
  problems.
\newblock {\em Manuf. Serv. Oper. Manag.}, 8(2):136--148, 2006.

\bibitem{MYERSON1983265}
Roger~B. Myerson and Mark~A. Satterthwaite.
\newblock Efficient mechanisms for bilateral trading.
\newblock {\em Journal of Economic Theory}, 29(2):265--281, 1983.

\bibitem{PW24}
Neel Patel and David Wajc.
\newblock Combinatorial stationary prophet inequalities.
\newblock In David~P. Woodruff, editor, {\em Proceedings of the 2024 {ACM-SIAM}
  Symposium on Discrete Algorithms, {SODA} 2024, Alexandria, VA, USA, January
  7-10, 2024}, pages 4605--4630. {SIAM}, 2024.

\bibitem{RCV25}
Surbhi Rajput, Ashish Chiplunkar, and Rohit Vaish.
\newblock Trading prophets: How to trade multiple stocks optimally.
\newblock In {\em Proceedings of the SIAM Symposium on Simplicity in Algorithms
  (SOSA25)}, pages 238--252, 2025.

\bibitem{R2020}
Tim Roughgarden, editor.
\newblock {\em Beyond the Worst-Case Analysis of Algorithms}.
\newblock Cambridge University Press, 2020.

\end{thebibliography}

\end{document}